\documentclass{llncs}
\usepackage{amssymb}
\usepackage{amsmath}
\usepackage{amsbsy}
\usepackage[utf8]{inputenc}
\usepackage[T1]{fontenc}
\usepackage{url}
\usepackage{graphicx}
\usepackage{caption}
\usepackage{verbatim}
\usepackage[usenames,dvipsnames]{xcolor}
\usepackage{float}
\usepackage[caption=false]{subfig}
\usepackage{multirow}
\usepackage{appendix}
\usepackage{thmtools,thm-restate}
\usepackage{soul}
\usepackage{hyperref}
\usepackage{color, colortbl}
\definecolor{LightCyan}{rgb}{0.88,1,1}
\makeatletter
\newcommand{\rd}{\delta}
\newcommand*\samethanks[1][\value{footnote}]{\footnotemark[#1]}
\makeatother
\title{Lagrangian Reachabililty}
\author{Jacek Cyranka\inst{1}\thanks{Both authors contributed equally to this work.}, 
Md. Ariful Islam\inst{2}\samethanks, 
Greg Byrne\inst{3}, Paul Jones\inst{4},\\ Scott A. Smolka\inst{5}, Radu Grosu\inst{6}}
\institute{  	
	Rutgers University
    \{\email{jcyranka@gmail.com}\}
    \and	
	Carnegie Mellon University
    \{\email{mdarifui@cs.cmu.edu}\}
    \and
	Metron Inc., Reston, VA
    \{\email{gr3g.byrne@gmail.com}\}
    \and
	US Food and Drug Administration
    \{\email{paul.jones@fda.hhs.gov}\}
    \and	
	Stony Brook University
    \{\email{sas@cs.stonybrook.edu}\}
    \and	
	Vienna University of Technology
    \{\email{radu.grosu@tuwien.ac.at}\}
}
\begin{document}
\maketitle
\vspace{-2ex}
\begin{abstract}
\label{sec:abstract}
We introduce LRT, a new Lagrangian-based ReachTube computation algorithm that conservatively approximates the set of reachable states of a nonlinear dynamical system. LRT makes use of the Cauchy-Green stretching factor (SF), which is derived from an over-approximation of the gradient of the solution-flows. The SF measures the discrepancy between two states propagated by the system solution from two initial states lying in a well-defined region, thereby allowing LRT to compute a reachtube with a ball-overestimate in a metric where the computed enclosure is as tight as possible.  To evaluate its performance, we implemented a prototype of LRT in C++/Matlab, and ran it on a set of well-established benchmarks. Our results show that LRT compares very favorably with respect to the CAPD and Flow* tools.
\end{abstract}
\section{Introduction}
\label{sec:intro}
Bounded-time reachability analysis is an essential technique for ensuring the safety of emerging systems, such as cyber-physical systems (CPS) and controlled biological systems (CBS). However, computing the reachable states of CPS and CBS is a very difficult task as these systems are most often nonlinear, and their state-space is uncountably infinite.  As such, these systems typically do not admit a closed-form solution that can be exploited during their analysis.

%One way to adress nonlinearity, is to (piecewise) approximate it in a linear fashion. Once this is achieved, one can exploit the reach set of available tools and techniques. They all take advantage of the known closed-form solution. A state of the art tool in this respect is SpaceEx \cite{spacex}.  The larger class of polynomial-affine systems is addressed, for example, in~\cite{ASB}, \cite{JA}.

For CPS/CBS, one can therefore only compute point solutions (trajectories) through numerical integration and for predefined inputs.  To cover the infinite set of states reachable by the system from an initial region, one needs to conservatively extend (symbolically surround) these pointwise solutions by enclosing them in \emph{reachtubes}. Moreover, the starting regions of these reachtubes have to cover the initial states region. 

The class of continuous dynamical systems we are interested in this paper are nonlinear, time-variant ordinary differential equations (ODEs): 
\begin{subequations}
\label{cauchy}
\begin{align}
  x^\prime(t)&= F( t, x(t) ),\\
  x(t_0)&= x_0,
\end{align}
\end{subequations}
\noindent{}where $x\colon\mathbb{R}\to\mathbb{R}^n$. We assume that $F$ is a smooth function, which guarantees short-time existence of solutions. The class of time-variant systems includes the class of time-invariant  systems. Time-variant equations may contain additional terms, e.g. an excitation variable, and/or periodic forcing terms.

For a given initial time $t_0$, set of initial states $\mathcal{X}\subset\mathbb{R}^n$, and time bound $T>t_0$, our goal is to compute  conservative \emph{reachtube} of \eqref{cauchy}, that is, a sequence of time-stamped sets of states $(R_1,t_1),\dots,(R_k,t_k)$ satisfying:
\[
\text{Reach}\left((t_0,\mathcal{X}\right),[t_{i-1},t_i]) \subset R_i\text{ for }i = 1,\dots,k,
\]
where $\text{Reach}\left((t_0,\mathcal{X}\right),[t_{i-1},t_i])$ denotes the set of reachable states of \eqref{cauchy} in the interval $[t_{i-1},t_i]$. Whereas there are many sets satisfying this requirement, of particular interest to us are reasonably tight reachtubes; i.e.\ reachtubes whose over-approximation 
%% over the true reach set 
is the tightest possible, having in mind the goal of proving that a certain region of the phase space is (un)safe, and avoiding false positives.
In practice and for the sake of comparision with other methods, we compute a discrete-time reachtube; as we discuss, a continuous reachtube can be obtained using our algorithm.
%There is a rich body of work for the computation reachtubes. Among the various techniques developed, one can mention the following: sensitivity analysis \cite{D,donze,dang,wholElse?}, optimization techniques~\cite{feinekos,sankaranaryan,strong?,whoelse?}, conservative numeric-bounds computation~\cite{CAPD}, \cite{Lo}, \cite{ZW1}, \cite{Z1}, \cite{CAS}, \cite{CAS2}, \cite{vnode}, \cite{BM}\cite{MB}, or discrepancy functions~\cite{Fan2015}, \cite{FKXS}. 

Existing tools and techniques for conservative reachtube computation
%%, most closely related to our work,
can be classified by the time-space approximation they perform into three categories: (1)~Taylor-expansion in time, variational-expansion in space (wrapping-effect reduction) of the solution set (CAPD~\cite{CAPD,ZW1,Z1}, VNode-LP~\cite{nedialkov2006vnode,vnode2}, (2)~Taylor-expansion in time and space of the solution set (Cosy Infinity~\cite{makino2006cosy,BM,MB}, Flow*~\cite{CAS,CAS2}), 
and (3)~Bloating-factor-based and discrepancy-function-based~\cite{FKXS,Fan2015}.  
The last technique computes a conservative reachtube using a \emph{discrepancy function} (DF) that is derived from an over-approximation of the \emph{Jacobian of the vector field} (usually given by the RHS of the differential equations) defining the continuous dynamical system.   
\begin{figure}
\vspace{-4ex}
\centering
\includegraphics[scale=0.6]{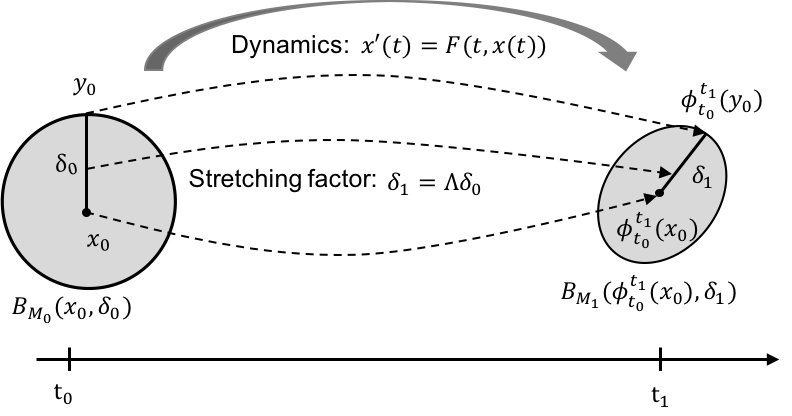}
\caption{An overview of LRT. The figure shows one execution step of the LRT described in detail in Section~\ref{sec:lrt}.  The dashed arrows reflect the solution flow $\phi$ and the evolution of state discrepancy.}
\label{fig:intro_fig}
\vspace{-2ex}
\end{figure}

This paper proposes an alternative (and orthogonal to~\cite{FKXS,Fan2015}) technique for computing a conservative reachtube, by using a \emph{stretching factor (SF)} that is derived from an over-approximation of the \emph{gradient of the solution-flows}, also referred to as the \emph{sensitivity matrix} \cite{breach,donze}, and the \emph{deformation tensor}. 
An illustration of our method is given in Figure~\ref{fig:intro_fig}. $B_{M_0}(x_0,\delta_0)$ is a well-defined initial region given as a ball in metric space $M_0$ centered at $x_0$ of radius $\delta_0$.  The SF $\Lambda$ measures the discrepancy of two states $x_0$, $y_0$ in $B_{M_0}$ propagated by the solution flow induced by \eqref{cauchy}, i.e.\ $\phi_{t_0}^{t_1}$. We can thus use the SF to bound the infinite set of reachable states at time $t_1$ with the ball-overestimate $B_{M_1}(\phi_{t_0}^{t_1}(x_0), \delta_1)$ in an appropriate metric (which may differ from the initial $M_0$), where $\delta_1 = \Lambda\cdot\delta_0$. Similar to~\cite{FKXS}, this metric is based on a weighted  norm, yielding a tightest-possible enclosure of the reach-set~\cite{D,L,LWS,MA}. For two-dimensional system, we present an analytical method to compute $M_1$, but for higher dimensional system, we solve a semi-definite optimization problem. Analytical formulas derived for 2d case allow for faster computation.
We point out that the output provided by LRT can be used to compute a validated bound for the so-called \emph{finite-time Lyapunov exponent} ($FTLE = \frac{1}{T}\ln(SF)$) for a whole set of solutions. FTLE are widely computed in e.g. climate research in order to detect Lagrangian coherent structures.
%
%\begin{comment}
%
%\end{comment}

We call this approach and its associated the LRT, for \emph{Lagrangian Reachtube} computation. The LRT uses analogues of Cauchy-Green deformation tensors (CGD) from finite strain theory (FST) to determine the SF of the solution-flows, after each of its time-step iterations. The LRT algorithm is described thoroughly in Section~\ref{sec:lrt}.

To compute the gradient of the flows, we currently make use of the CAPD $C^1$ routine, which propagates the initial ball (box) using interval arithmetic. The CAPD library has been certified to compute a conservative enclosure of the true solution, and it has been used in many peer-reviewed computer proofs of theorems in dynamical systems, including \cite{GZ,TZ,WZ2}.

To evaluate the LRT's performance, we implemented a prototype in C++/Ma\-tlab and ran it on a set of eight
benchmarks. 
%from~\cite{CAS2,FKXS}, the forced Van der Pol oscillator (time-variant system)~\cite{van1927vii}, and the  Mitchell Schaeffer cardiac cell model~\cite{mitchell03}.
%% In our evaluation, we were only interested in precision and accuracy, and not in the running time.  This is comparable to the one of existing tools. 
Our results show that the LRT compares very favorably to a direct use of CAPD and Flow* (see Section~\ref{sec:results}), while still leaving room for further improvement.  In general, we expect the LRT to behave favorably on systems that exhibit long-run stable behavior, such as orbital stability. 

We did not compare the LRT with the DF-based tools~\cite{Fan2015,FKXS}, although we would have liked to do this very much. The reason is that the publicly-available DF-prototype has not yet been certified to produce conservative results. Moreover, the prototype only considers time-invariant systems.

The rest of the paper is organized as follows. Section~\ref{sec:bgd} reviews finite-strain theory and the Cauchy-Green-deformation tensor for flows. Section~\ref{sec:lrt} presents the LRT, our main contribution, and proves that it conservatively computes the reachtube of a Cauchy system. Section~\ref{sec:results} compares our results to CAPD and Flow* on six benchmarks from~\cite{CAS2,FKXS}, the forced Van der Pol oscillator (time-variant system)~\cite{van1927vii}, and the  Mitchell Schaeffer cardiac cell model~\cite{mitchell03}. Section~\ref{sec:conclusions} offers our concluding remarks and discusses future work.

\section{Background on Flow Deformation}
\label{sec:bgd}
In this section we present some background on the LRT. First, in Section~\ref{seclagrangian} we briefly recall the general FST, as in the LRT we deal with matrices analogous to Cauchy-Green deformation tensors. Second, in Section~\ref{sec:cgd} we show how the Cauchy-Green deformation tensor can be used to measure discrepancy of two initial states propagated by the flow induced by Eq.~\eqref{cauchy}.
\subsection{Finite Strain Theory and Lagrangian Description of the Flow}
\label{seclagrangian}
In classical continuum mechanics, \emph{finite strain theory (FST)}  deals with the deformation of  a continuum body in which both rotation and strain can be arbitrarily large.
Changes in the configuration of the body are described by a displacement field. Displacement fields relate the initial configuration with the deformed configuration, which can be significantly different. FST has been applied, for example, in stress/deformation analysis of media like elastomers, plastically-deforming materials, and fluids, modeled by constitutive models (see e.g.,~\cite{HY}, and the references provided there). In the \emph{Lagrangian} 
representation, the coordinates describe the deformed configuration (in the material-reference-frame spatial coordinates), whereas in the \emph{Eulerian} representation,  the coordinates describe the undeformed configuration (in a fixed-reference-frame spatial coordinates).

\paragraph{Notation} In this section we use the standard notation used in the literature on FST. We use $X$ to denote the position of a particle in the Eulerian coordinates (undeformed), 
and $x$ to denote the position of a particle in the Lagrangian coordinates (deformed).
The Lagrangian coordinates depend on the initial (Eulerian) position, and the time $t$, so we use $x(X,t)$ to denote the position of a particle in Lagrangian coordinates.

The \emph{displacement field} from the initial configuration to the deformed configuration in Lagrangian coordinates is given by the following equation:
\begin{equation}
\label{displacement}
u(X,t) = x(X,t) - X.
\end{equation}
%\begin{comment}
The dependence $\nabla_X{u}$ of the displacement field $u(X,t)$ on the initial condition $X$ is called the \emph{material displacement gradient tensor}, with
\begin{equation}
\label{mdgt}
\nabla_X{u(X,t)} = \nabla_X{x}(X,t)-I,
\end{equation}
where $\nabla_X{x}$ is called the \emph{deformation gradient tensor}.

We now investigate how an initial perturbation $X+dX$ in the Eulerian coordinates evolves to the deformed configuration $dx(X+dX,t)$ in Lagrangian coordinates by using~\eqref{displacement}. This is called a \emph{relative displacement vector:}
%Let $dX$ be an initial perturbation. The perturbed initial position $X + dX$
%\[
%X + dX,
%\]
%in the Lagrangian coordinates, by using the displacement \eqref{displacement}, is
\[dx(X+dX,t)=x(X+dX,t)-x(X,t)=u(X+dX,t)+dX-u(X,t),\]
As a consequence, for small $dX$ we obtain the following approximate equality:
\begin{equation}
\label{amvt}
dx(X+dX,t) \approx u(X+dX,t)-u(X,t).
\end{equation}
Now let us compute $u(X+dX,t)$ by expressing $du(X+dX,t)$ as with $x(X+dX,t)$ above. One obtains:
\[
u(X+dX,t) = u(X,t) + du(X+dX,t) = 
%u(X,t) + \frac{\partial{u}(X+dX,t)}{\partial{X}}dX = 
u(X,t) + \nabla_X{u(X+dX,t)}dX.
\]
%Hence, $u(X+dX,t)$ is defined by the equation below:
%\[
%u(X+dX,t) = u(X,t) + \nabla_X{u(X,t)}dX.
%\]
Now by replacing $u(X+dX,t)$ in  Equation~(\ref{amvt}) above, one obtains the following result:
\[
dx(X+dX,t) \approx \nabla_X{x(X+dX,t)}dX.
\]
%\[
%dx(X+dX,t) = \nabla_X{x(X,t)}dX.
%\]
%\end{comment}
%radu: defined above: We call $\nabla_X{x} $ the \emph{deformation gradient} (DG).

Several rotation-independent tensors have been introduced in the literature. Classical examples include the right Cauchy-Green deformation tensor:
\begin{equation}
\label{cgtensor}
CG = \|\nabla_X{x}\|^2 = (\nabla_X{x})^T\cdot\nabla_X{x}.
\end{equation}
\subsection{Cauchy-Green Deformation Tensor for Flows}
\label{sec:cgd}
\paragraph{Notation}
By $[x]$ we denote a product of intervals (a box), i.e. a compact and connected set $[x]\subset\mathbb{R}^n$. We will use the same notation for 
interval matrices.
By $\|\cdot\|_2$ we denote the \emph{Euclidean norm}, by $\|\cdot\|_\infty$ we denote the max norm, we use the same notation for the induced operator norms. 
Let $B(x,\rd )$ denote the closed ball centered at $x$ with the radius $\rd$. It will be clear from the context in which metric space we consider the ball.
By $\phi_{t_0}^{t_1}$ we denote the flow induced by \eqref{cauchy}, by $D_x\phi_{t_0}^{t_1}$ we denote the partial derivative in $x$ of the flow with respect to the initial condition, at time $t_1$, which we call the \emph{gradient of the flow}, also referred to as the \emph{sensitivity matrix} \cite{breach,donze}. 

Let us now relate the finite strain theory presented in Section~\ref{seclagrangian} to the study of flows induced by the ODE~\eqref{cauchy}.
For expressing deformation in time of a continuum we first consider the set of initial conditions (e.g. a ball), which is being evolved (deformed) in time by the flow $\phi$.
For the case of flows we have that the positions in Eulerian coordinates are coordinates of the initial condition (denoted here using lower case letters with subscript $0$, i.e. $x_0$, $y_0$).
The equivalent of $x(X,t)$ -- the Lagrangian coordinates of $X$ at time $t$ -- is $\phi_{t_0}^t(x_0)$. 
Obviously, the equivalent of $u(X,t)$ is $(\phi_{t_0}^t(x_0) - x_0)$, and the deformation 
gradient $\nabla_X{x}$ here is just the derivative of the flow with respect to the initial condition $D_x\phi_{t_0}^t$ (sensitivity matrix).

\paragraph{}
In this section we show that deformation tensors arise in a study of discrepancy of solutions of \eqref{cauchy}. First, we provide some basic lemmas that we use in the analysis of our reachtube-computation algorithm. 
We work with the metric spaces that are based on weighted norms. We aim at finding weights, such that the induced weighted norms provides to be smaller than those in the Euclidean norm for the computed gradients of solutions. This procedure is similar to finding appropriate induced \emph{matrix measures} (also known as \emph{logarithmic norms}), which provides rate of expansion between two solutions, see \cite{D,L,LWS,MA}.
\begin{definition}
\label{defMnorm}
Given positive-definite symmetric matrix $M\in\mathbb{R}^{n\times n}$ we define the \emph{$M$-norm} of $\mathbb{R}^n$ vectors  by
\begin{equation}
\label{Mnorm}
\|y\|_M = \sqrt{y^TMy}.
\end{equation}
Given the decomposition
\[
M = C^TC,
\]
the matrix norm induced by \eqref{Mnorm} is
\begin{equation}
\label{MnormMatrix}
  \|A\|_M = \sqrt{\lambda_{max}\left((C^T)^{-1}\cdot A^T\cdot M\cdot A\cdot C^{-1}\right)},
\end{equation}
where $\lambda_{max}(\cdot)$ denotes the maximal eigenvalue.

Observe that the square-root is well defined, as $\lambda_{max}\,{\geq}\,0$, because $M$ is symmetric semi-positive definite, and hence, the matrix $(C^T)^{-1}\cdot A^T\cdot M\cdot A\cdot C^{-1}$ is also symmetric semi-positive definite.
\end{definition}

\begin{lemma}
\label{lemdiscrep}
Consider the Cauchy problem \eqref{cauchy}. Let $x_0,y_0\in\mathbb{R}^n$ be two initial conditions at time $t_0$.
Let $M\in\mathbb{R}^{n\times n}$ be a positive-definite symmetric matrix and $C^TC = M$ be its decomposition. For $t_1 \geq t_0$, it holds that 
\begin{equation}
\label{Mestim}
  \|\phi_{t_0}^{t_1}(x_0){-}\phi_{t_0}^{t_1}(y_0)\|_M\,{\leq}\,\sqrt{\lambda_{max}\left((C^T)^{-1}D_x\phi_{t_0}^{t_1}(\xi)^T M\,D_x\phi_{t_0}^{t_1}(\xi) C^{-1}\right)}~\|x_0-y_0\|_M
\end{equation}
where $\xi = \omega x_0 + (1-\omega)y_0$ for some $\omega\in[0,1]$.
For the particular case of the Euclidean norm, \eqref{Mestim} takes the form
\begin{equation}
\label{estim}
\|\phi_{t_0}^{t_1}(x_0)-\phi_{t_0}^{t_1}(y_0)\|_2 \leq \sqrt{\lambda_{max}\left(\left(D_x\phi_{t_0}^{t_1}(\xi)\right)^T\cdot D_x\phi_{t_0}^{t_1}(\xi)\right)}~\|x_0-y_0\|_2.
\end{equation}
\end{lemma}
%\begin{proof}
%
%Let $\xi(\omega) = \omega x_0 + (1 - \omega)y_0$.
%From 
%\[
%\int_0^1{D_x\phi_{t_0}^{t_1}( \xi(\omega) )\,d\omega} = \frac{1}{x_0-y_0}\left(\phi_{t_0}^{t_1}(x_0) - \phi_{t_0}^{t_1}(y_0)\right),
%\]
%and the well known mean value theorem for integrals, it holds that
%\begin{equation}
%\label{mvt}
%\phi_{t_0}^{t_1}(x_0)-\phi_{t_0}^{t_1}(y_0) = D_x\phi_{t_0}^{t_1}(\hat{\xi})(x_0-y_0)
%\end{equation}
%for some $\hat{\omega}\in[0,1]$, $\xi(\hat{\omega}) = \hat{\xi}$.
%We are interested in computing $\rd(T)$, i.e. an upper bound for the radius of the ball, such that it encloses all trajectories at time $T$
%\[
%\rd(T) \geq \sup_{y\in B(x_0,\rd(t_0))}{\|\phi_{t_0}^T(y) - \phi_{t_0}^T(x)\|}.
%\]
%From taking norms in \eqref{mvt} we obtain
%\[
%{\|\phi_{t_0}^{t_1}(x_0) - \phi_{t_0}^{t_1}(y_0)\|} = \|D_x\phi_{t_0}^{t_1}(\hat{\xi})(x_0-y_0)\|\leq \|D_x\phi_{t_0}^{t_1}(\hat{\xi})\|\|x_0-y_0\|.
%\]
%Replacing $\|D_x\phi_{t_0}^{t_1}(\hat{\xi})\|_2$ with the induced Euclidean matrix norm we obtain \eqref{estim}. If we use the weighted $M$-norm~\eqref{Mnorm}
%we have for the matrix norm \eqref{MnormMatrix}
%\begin{equation*}
%  \|\phi_{t_0}^{t_1}(x_0){-}\phi_{t_0}^{t_1}(y_0)\|_M\,{\leq}\,\sqrt{\lambda_{max}((C^T)^{-1}\left(D_x\phi_{t_0}^{t_1}(\hat{\xi})\right)^T 
%  M\,D_x\phi_{t_0}^{t_1}(\hat{\xi}) C^{-1})}~\|x_0-y_0\|_M
%\end{equation*}
%\qed
%\end{proof}
A proof can be found in Appendix~\ref{app:proofs}.
\begin{remark}
Let $\xi\in\mathbb{R}$ be a given vector. Observe that 
$
\left(D_x\phi_{t_0}^{t_1}(\xi)\right)^T\cdot D_x\phi_{t_0}^{t_1}(\xi)
$
appearing in \eqref{estim}
is the right \emph{Cauchy-Green deformation tensor} \eqref{cgtensor}
for two given initial vectors $x_0$ and $y_0$. We call the value $\sqrt{\lambda_{max}}$ appearing in \eqref{estim} \emph{Cauchy-Green stretching factor} for given initial vectors $x_0$ and $y_0$, 
which is necessarily positive as the CG deformation tensor is positive definite $\sqrt{\lambda_{max}(\cdot)} > 0$.
\end{remark}
Lemma~\ref{lemdiscrep}  is used when both of the discrepancy of the solutions at time $t_1$ as well as the initial conditions is measured in the same $M$-norm.
In the practical Lagrangian Reachtube Algorithm the norm used is changed during the computation. Hence we need another version of Lemma~\ref{lemdiscrep}, where the norm in which the discrepancy of the initial condition in measured differs from the norm in
which the discrepancy of the solutions at time $t_1$ is measured.
\begin{lemma}
\label{lemdiscrep2}
Consider the Cauchy problem \eqref{cauchy}. Let $x_0,y_0\in\mathbb{R}^n$ be two initial conditions at time $t_0$.
Let $M_0,M_1\in\mathbb{R}^{n\times n}$ be positive-definite symmetric matrices, and $C_0^TC_0 = M_0$, $C_1^TC_1 = M_1$ their decompositions
respectively. For $t_1 \geq t_0$, it holds that 
\begin{multline}
\label{cgm0m1}
\|\phi_{t_0}^{t_1}(x_0)-\phi_{t_0}^{t_1}(y_0)\|_{M_1}\\\leq \sqrt{\lambda_{max}\left((C_0^T)^{-1}\cdot\left(D_x\phi_{t_0}^{t_1}(\xi)\right)^T\cdot M_1\cdot D_x\phi_{t_0}^{t_1}(\xi)\cdot C_0^{-1}\right)}~\|x_0-y_0\|_{M_0},
\end{multline}
where $\xi = \omega x_0 + (1-\omega)y_0$ for some $\omega\in[0,1]$.
\end{lemma}
% \begin{proof}
% Let $\xi = \omega x_0 + (1-\omega)y_0$ for some $\omega\in[0,1]$. We use the equality $\phi_{t_0}^{t_1}(x_0)-\phi_{t_0}^{t_1}(y_0) = D_x\phi_{t_0}^{t_1}(\xi)(x_0-y_0)$ derived in the proof of Lemma~\ref{lemdiscrep}.
% %
% Let us denote $A:=D_x\phi_{t_0}^{t_1}(\xi)$, and $w = (x_0-y_0)$. It holds that
% \begin{multline*}
% \|\phi_{t_0}^{t_1}(x_0)-\phi_{t_0}^{t_1}(y_0)\|_{M_1} = \|Aw\|_{M_1}=\sqrt{(Aw)^TM_1(Aw)}=\sqrt{w^T(A^TM_1A)w}=\\
% \sqrt{w^TC_0^T((C_0^T)^{-1}A^TM_1AC_0^{-1})C_0w} \leq \sqrt{\lambda_{max}\left((C_0^T)^{-1}A^TM_1AC_0^{-1}\right)}\sqrt{w^TC_0^TC_0w}\\
% \hspace*{-5mm}=\sqrt{\lambda_{max}\left((C_0^T)^{-1}A^TM_1AC_0^{-1}\right)}~\|w\|_{M_0}\quad\qed
% \end{multline*}
% \end{proof}
%\bigskip
A proof can be found in Appendix~\ref{app:proofs}.
\begin{remark}
Given a positive-definite symmetric matrix $M$. We call the value appearing in \eqref{cgm0m1}
$
(C_0^T)^{-1}\cdot\left(D_x\phi_{t_0}^{t_1}(\xi)\right)^T\cdot M_1\cdot D_x\phi_{t_0}^{t_1}(\xi)\cdot C_0^{-1}
$
as the \emph{$M_0/M_1$-deformation tensor},
and the value $\sqrt{\lambda_{max}}$ as the $M_0/M_1$-\emph{stretching factor}.
\end{remark}
The idea behind using weighted norms in our approach is 
that the stretching factor in $M$-norm \eqref{cgm0m1} is expected to be smaller than that 
in the Euclidean norm \eqref{estim}.
Ultimately, this permits
a tighter reachtube computation, whose complete procedure is presented in Section~\ref{sec:rigcomp}.

%Finally, we obtain an upper bound for $\rd(T)$ -- the radius of the ball at time $T$  from the following inequalities
%\[
%\delta(T) \leq \sup_{y\in B(x,\rd(t_0))}{\|\phi_{t_0}^T(y) - \phi_{t_0}^T(x)\|_2}\leq \sqrt{\lambda_{max}([\Delta])}\|x-y\|_2\leq \sqrt{\lambda_{max}([\Delta])}\rd(t_0).
%\]
%\begin{remark}
%Observe that verifying $\sqrt{\lambda_{max}([\Delta])}<1$ shows that the radius of the ball at time $T$ is smaller that initially at time $t_0$,
%i.e. any pair of orbits with the i.c. within $B(x,\rd(t_0))$ gets attracted to each other at time $T$.
%This in turn means that the diameter of the reachtube in fact decreases in time, this is the property that could enable performing 
%a long-time integration. 
%\end{remark}

\section{Lagrangian Reachtube  Computation}
\label{sec:lrt}

\subsection{Reachtube Computation: Problem-Statement}
\label{sec:problem-statement}

In this section we provide first some lemmas that we then use to show that our method-and-algorithm produces a conservative output, in the sense that it encloses the set of solutions starting from a set of initial conditions. Precisely, we define what we mean by conservative enclosures. 
\begin{definition}
\label{defconserv}
Given an initial set $\mathcal{X}$, initial time $t_0$, and the target time $t_1\geq t_0$. We call the following compact sets:
\begin{itemize}
\item $\mathcal{W}\subset \mathbb{R}^n$ a \emph{conservative}, reach-set enclosure,  if $\phi_{t_0}^{t_1}(x)\in\mathcal{W}$ for all $x\in\mathcal{X}$.
\item $\mathcal{D}\subset \mathbb{R}^{n\times n}$ a \emph{conservative}, gradient enclosure, if $D_x\phi_{t_0}^{t_1}(x)\in\mathcal{D}$ for all $x\in\mathcal{X}$.
\end{itemize}
\end{definition}
\newcommand{\reach}{\mbox{Reach}}
\newcommand{\calX}{\mathcal{X}}
Following the notation used in \cite{FKXS}, and extending the corresponding definitions to our time variant setting, we introduce the notion of reachibility as follows: Given an initial set $\mathcal{X}\subset\mathbb{R}^n$ and a time $t_0$, we call a state $x$ in $\mathbb{R}^n$ as \emph{reachable} within a time interval $[t_1,t_2]$, if there exists an initial state $x_0\in\calX$ at time $t_0$ and a time  $t\in[t_1,t_2]$, such that $x=\phi_{t_0}^t(x_0)$. The set of all reachable states in the interval $[t_1,t_2]$ is called the \emph{reach set} and is denoted by $\reach((t_0,\calX),[t_1,t_2])$.

\begin{definition}[\cite{FKXS} Def.~2.4]
\label{defreachtube}
Given an initial set $\mathcal{X}$, initial time $t_0$, and a time bound $T$, a \emph{$((t_0,X), T)$-reachtube} of System~\eqref{cauchy} is a sequence of time-stamped sets 
$(R_1, t_1),\dots, (R_k, t_k)$ satisfying the following:
%\begin{enumerate}
 %\item 
 (1)~$t_0 \leq t_1 \leq \dots \leq t_k = T$,
  %\item 
  (2)~$\reach((t_0,\mathcal{X}), [t_{i-1}, t_i]) \subset R_i, \forall i = 1,\dots, k$.
%\end{enumerate}
\end{definition}

\begin{definition}
%Let $\mathcal{X}\subset\mathbb{R}^n$ be compact and connected, and $T>0$ be a time horizon.
%Let $\phi_{t_0}^t(x)$ be the solution of \eqref{cauchy} with the initial condition $(t_0,x)$ at time $t$, for $x\in\mathcal{X}$.
%
%We call the $(\mathcal{X},[t_0,T])$ \emph{reachtube over-approximation} of the flow defined by \eqref{cauchy} a time-parametrized set
%\[
%\mathcal{R}\colon[t_0,T]\to2^{\mathbb{R}^n},
%\]
%having the property that for all $t\in[t_0,T]$ and $x\in\mathcal{X}$ it holds that
%\[
%\phi_{t_0}^t(x)\in \mathcal{R}(t),
%\]
%where  $\mathcal{R}(t)$ is the cross-section of the reachtube over-approximation, in particular $\mathcal{R}({t_0})=\mathcal{X}$.
%We will also use the notation
%\[
%\mathcal{R}([t_1,t_2]),
%\]
%to denote the segment of the reachtube over-approximation for the time interval $[t_1,t_2]\subset\mathbb{R}.$.
%
Whenever the initial set $\mathcal{X}$ and the time horizon $T$ are known from the context we will skip the $(\mathcal{X},[t_0,T])$ part,
and will simply use the name (conservative) \emph{reachtube over-approximation} of the flow defined by \eqref{cauchy}.
\end{definition}
%
%\begin{definition}
%Let $\mathcal{X}\subset\mathbb{R}^n$ compact and connected, $T>0$ be a time horizon.
%Let $\phi_{t_0}^t(x)$ be the solution of \eqref{cauchy} with the initial condition $(t_0,x)$ at time $t$ for $x\in\mathcal{X}$.
%Let $D_x\phi_{t_0}^t(x)$ denote the gradient of the flow defined by \eqref{cauchy}, with the initial condition $(t_0,x)$.
%
%We call the $(\mathcal{X},[t_0,T])$ \emph{$C^1$ reachtube over-approximation} of the gradient of the flow defined by \eqref{cauchy} a time-parametrized set
%\[
%\mathcal{DR}\colon[t_0,T]\to 2^{\mathbb{R}^{n\times n}},
%\]
%having the property that for all $t\in[t_0,T]$ and $x\in\mathcal{X}$ it holds that
%\[
%D_x\phi_{t_0}^t(x)\in \mathcal{DR}(t),
%\]
%where $\mathcal{DR}(t)$ is the cross-section of the $C^1$ reachtube over-approximation, in particular 
%$\mathcal{DR}({t_0}) = \{Id\}$. 
%We will also use the notation $\mathcal{DR}({[t_1,t_2]})$, as in the previous definition,
%to denote the segment of the reachtube over-approximation for the time interval $[t_1,t_2]\subset\mathbb{R}.$.
%
%Whenever the initial set $\mathcal{X}$ and the time horizon $T$ is known from the context we will skip $(\mathcal{X},[t_0,T])$ part,
%and use simply the name $C^1$ \emph{reachtube over-approximation} of the flow defined by \eqref{cauchy}.
%\end{definition}
%
Observe that we do not address here the question what is the exact structure of the solution set at time $t$ initiating at $\mathcal{X}$. In general it could have, for instance, a fractal structure. We aim at constructing an over-approximation for the solution set, and its gradient,
which is amenable for rigorous numeric computations.

In the theorems below we show that the method presented in this paper can be used to construct a reachtube over-approximation $\mathcal{R}$.
The theorems below is the foundation of our novel \emph{Lagrangian Reachtube Algorithm (LRT)} presented in Section~\ref{sec:rigcomp}.
In particular, the theorems below provide estimates we use for constructing  $\mathcal{R}$. First, we present a theorem for the discrete case.
\subsection{Conservative Reachtube Construction}
\label{sec:crt}

\begin{theorem}
\label{thmmaindiscrete}
Let $t_0\leq t_1$ be two time points.
Let $\phi_{t_0}^{t_1}(x)$ be the solution of \eqref{cauchy} with the initial condition $(t_0,x)$ at time $t_1$,
let $D_x\phi_{t_0}^{t_1}$ be the gradient of the flow.
Let $M_0,M_1\in\mathbb{R}^{n\times n}$ be positive-definite symmetric matrices, and $C_0^TC_0 = M_0$, $C_1^TC_1 = M_1$ be their decompositions
respectively.
Let $\mathcal{X} = B_{M_0}(x_0,\delta_0)\subset\mathbb{R}^n$ be a set of initial states for the Cauchy problem \eqref{cauchy}
(ball in $M_0$-norm with the center at $x_0$, and radius $\delta_0$). Assume that there exists a compact, 
conservative enclosure $\mathcal{D}\subset\mathbb{R}^{n\times n}$ for the gradients, such that:
\begin{equation}
\label{D}
D_x\phi_{t_0}^{t_1}(x)\in \mathcal{D}\text{ for all }x\in\mathcal{X}.
\end{equation}

Suppose $\Lambda>0$ is  such that:
\begin{equation}
\label{lambdageq}
\Lambda \geq \sqrt{ \lambda_{max}\left( (C_0^T)^{-1}D^TM_1DC^{-1}_0 \right) },\text{ for all }D\in\mathcal{D}.
\end{equation}

Then it holds that:
\[
\phi_{t_0}^{t_1}(x)\in B_{M_1}(\phi_{t_0}^{t_1}(x_0), \Lambda \cdot \delta_0).
\]

%OLD VERSION
%Let $\mathcal{X} = B(x_0,\delta_0)\subset\mathbb{R}^n$ be a set of initial conditions for the Cauchy problem \eqref{cauchy},
%$x\in\mathcal{X}$, $t\in [t_0,t_1]$. 
%Let $\phi_{t_0}^t(x)$ be the solution of \eqref{cauchy} with the initial condition $(t_0,x)$ at time $t\in[t_0,t_1]$. 
%Assume the first segment of a \emph{$C^1$ reachtube over-approximation} of the gradient of the flow defined by \eqref{cauchy} is given by 
%$\mathcal{DR}({[t_0,t_1]})$.
%
%et 
%\begin{gather*}
%\Delta\colon[t_0,t_1]\to 2^{\mathbb{R}^{n\times n}},\\
%\Delta(t) := \left\{ g^*\cdot g,\text{ for all }g\in\mathcal{DR}({t})\right\},\text{ for all times }t\in[t_0,t_1]
%\end{gather*}
% be the set of Cauchy-Green deformation tensors computed for all gradients contained in 
%$\mathcal{DR}({[t_0,t_1]})$.
%
%Assume there exists $\Lambda\colon[t_0,t_1]\to\mathbb{R}_+$, such that for all $t\in[t_0,t_1]$ it holds that
%\[
%\lambda_{max}{(\hat{\Delta})}  \leq \Lambda(t)\text{ for all }\hat{\Delta} \in \Delta(t).
%\]
%
%Then, for any $x\in\mathcal{X}$, $t\in[t_0,t_1]$ it holds that
%\[
%\phi_{t_0}^t(x) \in B\left(\phi_{t_0}^t(x_0), \delta_0\sqrt{\Lambda(t)}\right),
%\]
%Hence, $\mathcal{R}({[t_0,t_1]})$ -- the first segment of reachtube over-approximation is given by 
%\begin{equation*}
%mathcal{R}({t}) := B\left(\phi_{0}^{t}(x_0), \delta_0\sqrt{\Lambda(t)}\right)\text{ for }t\in[t_0,t_{1}].
%\end{equation*}
%and, for any $x,y\in\mathcal{X}$
%\[
%\|\phi_{t_0}^t(x) - \phi_{t_0}^t(y)\| \leq e^{\log{\delta_0\sqrt{\lambda_{max}(t)}}}\|x-y\|.
%\]
\end{theorem}
\begin{proof}
Let $x_0$ be the center of the ball of initial conditions $\mathcal{X} = B_{M_0}(x_0,\delta_0)$, and
let us pick $x \in \mathcal{X}$. From Lemma~\ref{lemdiscrep2} the discrepancy of the solutions
initiating at $x_0$ and $x$ at time $t_1$ is bounded in $M_1$-norm by:
\[
\|\phi_{t_0}^{t_1}(x_0)-\phi_{t_0}^{t_1}(x)\|_{M_1} \leq \delta_0\sqrt{\lambda_{max}\left((C_0^T)^{-1}\cdot\left(D_x\phi_{t_0}^{t_1}(\xi)\right)^T\cdot M_1\cdot D_x\phi_{t_0}^{t_1}(\xi)\cdot C_0^{-1}\right)},
\]
where $\xi = \omega x_0 + (1-\omega)x$ for some $\omega\in[0,1]$. Obviously, $\xi\in B_{M_0}(x_0,\delta_0)$.
Hence, 
%\[
$D_x\phi_{t_0}^{t_1}(\xi)\in\mathcal{D}.
$
%\]
Moreover, if $\Lambda>0$ satisfies \eqref{lambdageq}, then
\[
\Lambda \geq \sqrt{\lambda_{max}\left((C_0^T)^{-1}\cdot\left(D_x\phi_{t_0}^{t_1}(\xi)\right)^T\cdot M_1\cdot D_x\phi_{t_0}^{t_1}(\xi)\cdot C_0^{-1}\right)},
\]
and
%\[
$\phi_{t_0}^{t_1}(x) \in B_{M_1}(\phi_{t_0}^{t_1}(x_0), \Lambda \delta_0).
$
%\]
As $x$ was chosen in an arbitrary way, we are done.
\qed\end{proof}
\bigskip

The next theorem is the variant of Theorem~\ref{thmmaindiscrete} for obtaining a continuous reachtube.
\begin{theorem}
\label{thmmaincont}
Let $\phi_{t_0}^{t_1}(x)$ be the solution of \eqref{cauchy} with the initial condition $(t_0,x)$ at time $t_1$,
let $D_x\phi_{t_0}^{t_1}$ be the gradient of the flow.
Let $M_0,M_1\in\mathbb{R}^{n\times n}$ be positive definite symmetric matrices, and $C_0^TC_0 = M_0$, $C_1^TC_1 = M_1$ be their decompositions
respectively.
Let $\mathcal{X} = B_{M_0}(x_0,\delta_0)\subset\mathbb{R}^n$ be a set of initial conditions for the Cauchy problem \eqref{cauchy}
(ball in $M_0$-norm). Assume that there exists $\{\mathcal{D}_t\}_{t\in[t_0,t_1]}$ -- a compact $t$-parametrized set, 
such that
\begin{gather*}
\mathcal{D}_t\subset\mathbb{R}^{n\times n}\text{ for }t\in[t_0,t_1],\\
D_x\phi_{t_0}^{t}(x)\in \mathcal{D}_{t}\text{ for all }x\in\mathcal{X}\text{, and }t\in[t_0,t_1].
\end{gather*}

If $\Lambda > 0$ is such that 
%\begin{gather*}
$\text{for all }D\colon D\in\mathcal{D}_t\text{ for some }t\in[t_0,t_1]$
\[
\Lambda \geq \sqrt{ \lambda_{max}\left( (C_0^T)^{-1}D^TM_1DC^{-1}_0 \right) }
\]
%\end{gather*}

Then for all $t\in[t_0,t_1]$ it holds that
\begin{equation}
\label{reachset}
\phi_{t_0}^{t}(x)\in B_{M_1}(\phi_{t_0}^{t}(x_0), \Lambda \cdot \delta_0).
\end{equation}
\end{theorem}
\label{thmcontrt}
\begin{proof} It follows from the proof of Theorem~\ref{thmmaindiscrete} applied to all times $t\in[t_0,t_1]$.
\qed\end{proof}
\begin{corollary}
Let $T \geq t_0$. 
Assume that that there exists $\{\mathcal{D}_t\}_{t\in[t_0,T]}$ -- a compact $t$-parametrized set, 
such that
\begin{gather*}
\mathcal{D}_t\subset\mathbb{R}^{n\times n}\text{ for }t\in[t_0,T],\\
D_x\phi_{t_0}^{t}(x)\in \mathcal{D}_{t}\text{ for all }x\in\mathcal{X}\text{, and }t\in[t_0,T].
\end{gather*}

Then the existence of a  $((t_0,X), T)$-reachtube of the system described in Eq. \eqref{cauchy} 
in sense of Definition~\ref{defreachtube}, i.e. a sequence of time-stamped sets 
$(R_1, t_1),\dots, (R_k, t_k)$ is provided by an application of 
Lemma~\ref{thmmaincont}.
We provide an algorithm computing the reachtube in Section~\ref{sec:rigcomp}.
\end{corollary}
\begin{proof} Immediate application of Theorem~\ref{thmmaincont} shows that if the first segment $(R_1, t_1)$ 
is defined
\[
(R_1, t_1) := \bigcup_{t\in[t_0,t_1]}{B_{M_1}(\phi_{t_0}^{t}(x_0), \Lambda \cdot \delta_0)},
\]
then it satisfies
\[
\reach((t_0,\mathcal{X}), [t_{0}, t_1]) \subset (R_1, t_1),
\]
which is exactly provided by \eqref{reachset}.
The $j$-th segment $(R_j, t_j)$ for $j = 2,\dots,k$ is obtained by replacing in Theorem~\ref{thmmaincont} the time interval $[t_0,t_1]$
with the interval $[t_{j-1}, t_j]$ (Observe that the norm may be different in each step).
\qed\end{proof}
\subsection{LRT: A Rigorous Lagrangian Computation Algorithm}
\label{sec:rigcomp}
We now present a complete description of our algorithm. In the next section %subsection~\ref{secalgcorrect} 
we prove its correctness.
First let us comment on how in practice we compute a representable enclosure for the gradients \eqref{D}. 

\begin{itemize}
\vspace*{-1mm}\item First, given an initial 
ball $B_{M_0}(x_0,\delta_0)$ we compute its representable over-approximation, i.e.,~a product of intervals (a box in canonical coordinates $[X]\subset\mathbb{R}^n$), such that 
%\[
$
B_{M_0}(x_0,\delta_0) \subset [X]
$.
%\]
%
\vspace*{2mm}\item Next, using the \emph{$C^1$-CAPD algorithm} \cite{Z1,ZW1}
all trajectories initiating in $[X]$ are rigorously propagated forward in time, in order to compute a conservative enclosure for 
%\[
$
\left\{D_x\phi_{t_0}^{t_1}(\xi)|\ \xi\in [X]\right\},
$
%\]
%
%\paragraph{Notation} 
the gradients. We  use the notation
%\begin{equation}
%\label{gradenc}
$
[D_x\phi_{t_0}^{t_1}([X])] \subset \mathbb{R}^{n\times n},
$
%\end{equation}
to denote a representable enclosure (an interval matrix) for the set of gradients
%\[
$
\left\{D_x\phi_{t_0}^{t_1}(\xi)|\ \xi\in [X]\right\}.
$
%\]
\end{itemize}

\vspace*{-1mm}The norm of an interval vector $\|[x]\|$ is defined as the supremum of the norms of all vectors within the bounds $[x]$. 
For an interval set $[x]\subset\mathbb{R}^n$ we denote by $\mathbb{R}^n\supset B_{M}([x], r):= \bigcup_{x\in[x]}{B_M(x, r)}$, the union of the balls in $M$-norm 
of radius $r$ having the center in $[x]$.
Each product of two interval matrices is overestimated by using the interval-arithmetic operations.

%We are going to use the following notation for interval matrix enclosures of the Cauchy-Green 
%deformation tensors based on the gradient enclosure \eqref{gradenc}, computed in Euclidean and the $M$-norm respectively.
%\[
%[\Delta] = \left[ [D_x\phi_{t_0}^{t_1}([X])\right]^T\cdot \left[D_x\phi_{t_0}^{t_1}([X])\right],
%\]
%and equivalently
%\[
%[\Delta_M] = \left[(C^T)^{-1}\cdot\left[D_x\phi_{t_0}^{t_1}([X])\right]^T\cdot M\cdot \left[D_x\phi_{t_0}^{t_1}([X])\right]\cdot C^{-1}\right].
%\]
%
\begin{definition}
We will call the $C^1$-CAPD algorithm the \emph{rigorous tool}, which is currently used to generate conservative enclosures for the gradient  in the LRT.
\end{definition}

The output of the LRT are 
discrete-time reachtube over-approxima\-tion cross-sections 
 $\{t_0,t_1,\dots,t_k\}$, $t_k\,{=}\,t_0\,{+}\,kT$, i.e., reachtube over-approxima\-tions $B_{M_j}(x_j, \delta_j)$ of the flow induced by \eqref{cauchy} at time $t_j$. We note that the algorithm can be easily modified to provide a validated bounds for the finite-time Laypunov exponent.
We use the discrete-time output for the sake of comparison.  However, as a byproduct, a continuous reachtube
over-approximation is obtained by means of \emph{rough enclosures} (Fig.~\ref{enclosure}) produced by the rigorous tool used and by applying Theorem~\ref{thmmaincont}.
The implementation details of the algorithm can be found in Section~\ref{secimpldet}.
\begin{figure}
\vspace*{-3ex}
\begin{center}
\includegraphics[width=0.65\textwidth]{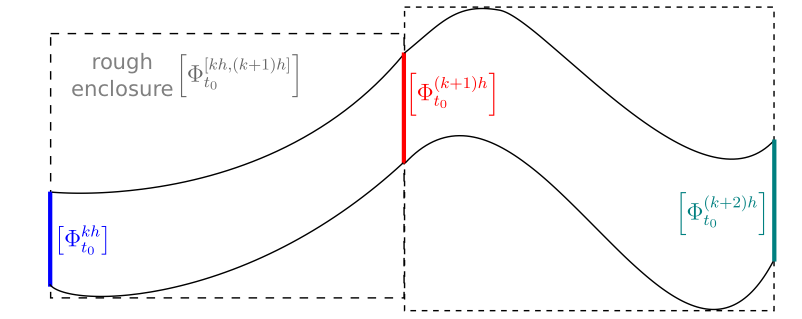}
\end{center}
\vspace*{-2ex}
\caption{Fine bounds provided at equally spaced time steps (colored), and coarse bounds provided for the intermediate times by the rough enclosure (dotted boxes).}
\vspace*{-3ex}
\label{enclosure}
\end{figure}
%
%We present rough enclosures graphically  on Figure~\ref{enclosure}. As rough enclosures are enclosures valid
%over whole time intervals, they necessary contain enclosures at time steps $\{t_0,t_1,\dots,t_k\}$. 
%Still, the algorithms implemented by the rigorous tools are optimized such that the enclosures at time steps
%are tighter (they are valid only for the given time step) contrary to the rough enclosures (valid 
%over whole time intervals). 
%

We are now ready to give the formal description of the LRT: (1)~its inputs, (2)~its outputs, and (3)~its computation.\\[2ex]

\noindent{\bf LRT: The Lagrangian Reachtube Algorithm}\\[1ex]
{\bf Input:}
\begin{itemize}
\vspace*{-1mm}\item \emph{ODE}s~\eqref{cauchy}: time-variant ordinary differential equations,
\item $T$: time horizon, $t_0$: initial time, 
%(for TI equations $t_0 = 0$), 
$k$: number of steps, $h\,{=}\,T{/}k$: time step
\item $M_0$: initial positive-definite symmetric  matrix defining the norm \eqref{Mnorm}, 
\item $[x_0]\subset \mathbb{R}^n$: the initial bounds for the position of the center of the ball at $t_0$, %(can be an interval),
\item $\rd_0 > 0$: the radius of the ball (in $M_0$ norm) about $x_0$ at $t_0$.\\[-4mm]
\end{itemize}
\noindent{\bf Output:}
\begin{itemize}
\vspace*{-1mm}\item $\{ [x_j] \}_{j=1}^k\,{\subset}\,\mathbb{R}^{n\times{k}}$: interval enclosures for ball centers $x_j$ at time 
$t_0\,{+}\,jh$, 
% radu: does not make sense?? where $t_k=t_0 + T$,
\item $\{ M_j \}_{j=1}^k$: norms 
%(see Definition~\ref{defMnorm}) 
defining metric spaces for the ball enclosures,
\item $\{ \rd_j \}_{j=1}^k\,{\in}\,\mathbb{R}_+^{k}$: radii of the ball enclosures at $x_j$ for $j = 1,\dots,k$. \footnote{Observe that the radius is valid for the $M_j$ norm,
$B_{M_j}([x_j],\delta_j)\subset\mathbb{R}^n$ for $j = 1,\dots,k$ is a conservative output, i.e.
$B_{M_j}([x_j],\delta_j)$ is an over-approximation for the set of states reachable at time $t_1$ starting from any state $(t_0,x)$, such that $x\in \mathcal{X}$
\[
\reach((t_0,\mathcal{X}), t_j) \subset 
B_{M_j}([x_j], \delta_j)\text{, for } j = 1,\dots,k.
\]}
\end{itemize}
\vspace*{1mm}\noindent{\bf Begin LRT}
\begin{enumerate}
\vspace*{-1mm}\item Set $t_1\,{=}\,t_0\,{+}\,h$.
Propagate the center of the ball $[x_0]$ forward in time by the time-step $h$, using the rigorous tool. The result is
a conservative enclosure for the solutions $[\phi_{t_0}^{t_1}([x_0])]$ and the gradients $[D_x\phi_{t_0}^{t_1}([x_0])]$.
\item Choose a matrix $D\in [D_x\phi_{t_0}^{t_1}([x_0])]$,
and compute a symmetric positive-definite matrix $M_1$, and its decomposition $M_1 = C_1^TC_1$, such that, it
minimizes the stretching factor for $D$. In other words, it holds that:
\begin{equation}
\label{Mopt}
\sqrt{\lambda_{max}\left( (C_1^T)^{-1}D^TM_1DC_1^{-1} \right)} \leq 
\sqrt{\lambda_{max}\left( (\widetilde{C}^T)^{-1}D^T\widetilde{M}D\widetilde{C}^{-1} \right)},
\end{equation}
for all positive-definite symmetric matrices $\widetilde{M}$. In the actual code we find the minimum with some resolution, i.e.,~we compute  $M_1$, such that it is close to the optimal in the sense of \eqref{Mopt},
using the procedure presented in subsection~\ref{secoptnorm}.
\item Decide whether to change the norm of the ball enclosure from $M_0$ to $M_1$ (if it leads to a smaller stretching factor). If the norm is to be changed keep $M_1$ as it is, otherwise $M_1 = M_0$,
\item Compute an over-approximation for $B_{M_0}([x_0],\delta_0)$, which is representable in the rigorous tool employed by the LRT, and can be used as input to propagate forward in time all solutions initiating in $B_{M_0}([x_0],\delta_0)$. This is a product of intervals in canonical coordinates $[X]\subset\mathbb{R}^n$, such that: 
\[
B_{M_0}([x_0],\delta_0) \subset [X].
\]
We compute the over-approximation using the interval arithmetic expression:
\[
{C_0^{-1}(C_0\cdot[x_0] + [-\delta_0,\delta_0]^n)}
\]
\item Rigorously propagate $[X]$ forward in time, using the rigorous tool over the time interval $[t_0, t_1]$. The result is a continuous reachtube, providing bounds for $[\phi_{t_0}^t([X])]$, and $[D_x\phi_{t_0}^t([X])]$ for all $t\in[t_0, t_1]$. We employ an integration algorithm with a fixed time-step $h$. As a consequence \emph{"fine"} bounds are obtained for $t=t_1$.
We denote those bounds by: 
\[
[\phi_{t_0}^{t_1}([X])], \quad [D_x\phi_{t_0}^{t_1}([X])].
\]
We remark that for the intermediate time-bounds, i.e.,~for: 
\begin{equation}
\label{roughenc}
[\phi_{t_0}^t([X])]\text{, for }t\in\left(t_0, t_1\right),
\end{equation}
the so-called \emph{rough enclosures} can be used. These provide \emph{coarse} bounds, as graphically illustrated in Figure~\ref{enclosure} in the appendix.

\item Compute interval matrix bounds for the $M_0/M_1$ Cauchy-Green deformation tensors:
\begin{equation}
\label{cgint}
\left[\left((C^T_0)^{-1}\cdot\left[D_x\phi_{t_0}^{t_1}([X])\right]^T\cdot C_1^T\right)\cdot\left(C_1\cdot \left[D_x\phi_{t_0}^{t_1}([X])\right]\cdot C^{-1}_0\right)\right],
\end{equation}
where $C^T_0, C_0$ are s.t.~$C^T_0C_0 = M_0$, and $C^T_1, C_1$ are s.t.~$C^T_1C^T = M_1$. The interval matrix operations are executed in the order given by the brackets.
\item Compute a value $\Lambda>0$ ($M_0/M_1$ stretching factor) as an upper bound for the square-root of the maximal eigenvalue of each (symmetric) matrix in \eqref{cgint}:
\begin{gather*}
\Lambda \geq \sqrt{\lambda_{max}\left( C \right)},\\\text{for all }C\in\left[\left((C^T_0)^{-1}\cdot\left[D_x\phi_{t_0}^{t_1}([X])\right]^T\cdot C_1^T\right)\cdot\left(C_1\cdot\left[D_x\phi_{t_0}^{t_1}([X])\right]\cdot C^{-1}_0\right)\right].
\end{gather*}
This quantity $\Lambda$ can be used for the purpose of computation of validated bound for the \emph{finite-time Lyapunov exponent}
\[
FTLE(B_{M_0}([x_0],\delta_0)) = \frac{1}{t_1-t_0}\ln(\Lambda_{prev}\cdot\Lambda),
\]
where $\Lambda_{prev}$ is the product of all stretching factors computed in previous steps.
\item Compute the new radius for the ball at time $t_1\,{=}\,t_0\,{+}\,h$:
\[
\rd(t_1) = \Lambda\cdot\rd(t_0),
\]
\item Set the new center of the ball at time $t_1$ as follows:
\[
[x_{1}]=[\phi_{t_0}^{t_1}([x_0])].
\]
%
%\item If in Step 5 a suitable matrix $M$ was found, and used to compute the norm and the radius of the ball $\delta_M(t)$ in Step~6, compute a CAPD representable enclosure for the $M$-ball of radius $\delta_M(t)$.
%
\item Set the initial time to $t_1$, the bounds for the initial  center of the ball to $[x_{1}]$, the current norm to $M_1$, the radius in $M_1$-norm to $\delta_1\,{=}\,\Lambda\cdot\rd_0$, and the ball enclosure for the set of initial states
to $B_{M_1}([x_{1}], \Lambda\cdot\rd(t_0))$. 
If $t_1 \geq T$ terminate. Otherwise go back to 1. 
\end{enumerate}
\noindent{\bf End LRT}

\vspace*{-1ex}
\subsection{LRT-Algorithm Correctness Proof}
\label{secalgcorrect}
In this section we provide a proof that the LRT, our new reachtube-computation algorithm, is an overapproximation of the behavior of the system described by Equations~\eqref{cauchy}. This main result is captured by the following theorem.

\begin{theorem}[LRT-Conservativity]
Assume that the rigorous tool used in the Lagrangian Reachtube Algorithm (LRT) produces
conservative gradient enclosures for 
system~\eqref{cauchy} in the sense of Definition~\ref{defconserv}, and it guarantees the existence of the solutions within time intervals. Assume also that the LRT terminates on the provided inputs. 

Then,
the output of the LRT is a \emph{conservative reachtube over-approximation of~\eqref{cauchy}} at times $\{t_j\}_{j=0}^k$, that is:
\[
\reach((t_0,\mathcal{X}), t_j) \subset 
B_{M_j}([x_j], \delta_j)\text{, for } j = 1,\dots,k,
\]
bounded solutions exists for all intermediate times $t\in(t_j,t_{j+1})$.
\end{theorem}
\begin{proof}
Let $\mathcal{X} = B_{M_0}([x_0],\delta_0)$ be a ball enclosure for the set of initial states. Without loosing generality we analyze the first step of the algorithm. The same argument applies to the consecutive steps (with the initial condition changed appropriately,
as explained in the last step of the algorithm).

The representable enclosure  $[X]\subset\mathbb{R}^n$ (product of intervals in canonical coordinates)  computed in Step~4 
satisfies $B_{M_0}([x_0],\delta_0)\subset [X]$. 
By the assumption the rigorous-tool used produces conservative  enclosures for the gradient of the flow 
induced by ODEs \eqref{cauchy}. Hence, as the set $[X]$ containing $B_{M_0}([x_0],\delta_0)$
is the input to the rigorous forward-time integration procedure in Step~5,
for the gradient enclosure $[D_x\phi_{t_0}^{t_1}([X])]$ computed in Step~5, it holds that 
\[
\left\{D_x\phi_{t_0}^{t_1}(x)\text{ for all }x\in\mathcal{X}\right\} \subset [D_x\phi_{t_0}^{t_1}([X])].
\]
Therefore, the set $[D_x\phi_{t_0}^{t_1}([X])]$ can be interpreted as the set $\mathcal{D}$, i.e., the compact set containing all gradients of the solution at time $t_1$ with initial condition in $B_{M_0}([x_0],\delta_0)$ appearing in Theorem~\ref{thmmaindiscrete}, see \eqref{D}.
As a consequence, the value $\Lambda$ computed in Step~7 satisfies the following inequality:
\[
\Lambda \geq \sqrt{ \lambda_{max}\left( (C_1^T)^{-1}D^TM_0DC^{-1}_1 \right) },\text{ for all }D\in\mathcal{D}.
\]
From Theorem~\ref{thmmaincont} it follows that: 
\[
\phi_{t_0}^{t_1}(x)\in B_{M_1}(\phi_{t_0}^{t_1}([x_0]), \Lambda \cdot \delta(t_0))\text{ for all }x\in\mathcal{X},
\]
which implies that:
\[
\reach((t_0,\mathcal{X}), t_1)\subset B_{M_1}( [\phi_{t_0}^{t_1}([x_0])], \Lambda\cdot\rd(t_0))
\]
 which proves our overapproximation (conservativity) claim. Existence of the solutions for all times $t\in(t_0,t_1)$ is guaranteed by the assumption about the rigorous tool (in the step~5 of the LRT algorithm we use rough enclosures, see Figure~\ref{enclosure}).
\qed
\end{proof}
%
%\subsection{Means of computing an over-approximation for the gradient of the flow using tools other than CAPD (e.g. Flow*)}
%\label{secgradcomput}
%
%\textbf{JC write that Flow* provides dense output which allows to compute continuous reachtube}
%
%
\subsection{Wrapping Effect in the Algorithm}
\label{sec:wrapping}

A very important precision-loss issue of conservative approximations, and therefore of validated methods for ODEs, 
is the \emph{wrapping effect}. This occurs when a set of states is wrapped (conservatively enclosed) within a box defined in a particular norm. The weighted $M$-norms technique the LRT uses (instead of the standard Euclidean norm) is a way of reducing the effect of wrapping. 
More precisely, in Step~7 of every iteration, the LRT finds an appropriate norm which minimizes the stretching factor computed from the set of Cauchy-Green deformation tensors. 

However, there are other sources of the wrapping effect in the algorithm.
The discrete reachtube bounds are in a form of a ball in appropriate metric space, which is an ellipsoidal set in canonical coordinates, see example on Fig.~\ref{fig:ss-level}. 
In Step~4 of the LRT algorithm, a representable enclosure in canonical coordinates for the ellipsoidal reachtube over-approximation is computed. When the ellipsoidal set is being directly wrapped into a box in canonical coordinates (how it is done in the algorithm currently), the wrapping effect is considerably larger than when the ellipsoidal set is wrapped into a rectangular set reflecting the eigen-coordinates.
We illustrate the wrapping effect using the following weighted norm (taken from one of our experiments): 
\begin{equation}
\label{M}
M = 
\begin{bmatrix}
7&-9.5\\
-9.5&19 
\end{bmatrix}
\end{equation}

\noindent
Figure~\ref{fig} shows the computation of enclosures for a ball represented in the weighted norm given by $M$. It is clear from Figure~1(c) that the box enclosure of the ball in the eigen-coordinates (blue rectangle) is much tighter than the box enclosure of it in the canonical coordinates (green square).
\begin{figure}[H] 
\vspace*{-0.2in}
\centering
\subfloat[]
{\includegraphics[width=0.3\linewidth]{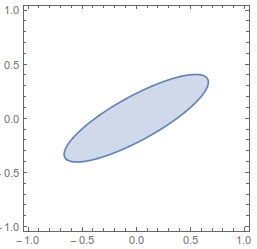} 
\label{fig:ss-level}}
\subfloat []
{\includegraphics[width=0.3\linewidth]{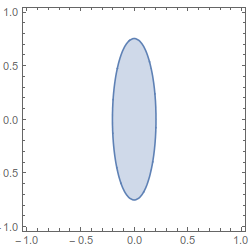}}
\subfloat[]
{{\includegraphics[width=0.31\linewidth]{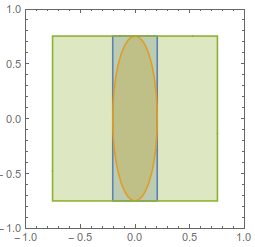} 
}\label{figc}}
\caption{(a)~A ball in the weighted norm given by $M$ of radius $1$ (the ellipsoidal set).
(b)~The ellipsoidal set in its eigen-coordinates (unrotated).
(c)~Wrapping the ellipsoidal set in a box: blue rectangle in eigen-coordinates and green square in canonical-coordinates.
}
\label{fig}
\vspace*{-.2in}
\end{figure}

Step~6 of the LRT is another place where reducing the wrapping effect has the potential to considerably increase the precision of the LRT. This step computes the product of interval matrices, which results in large overapproximations for wide-intervals matrices. In fact, in the experiments considered in Section~\ref{sec:results}, if the initial-ball radius is large, we observed that the overestimate of the stretching factor tends to worsen the LRT performance in reachtube construction, when compared to a direct application of CAPD and Flow*. We plan to find workarounds for this problem. One possible solution would be to use matrix decomposition, and compute the eigenvalues of the matrix by using this decomposition.
\subsection{Direct  computation of the optimal $M$-norm} 
%{for estimation of the stretching factor}
\label{secoptnorm}

The computation of the optimal $M$-norm enables the estimation of the streching factor. Step~7 of the LRT finds norm $M_1$ and decomposes it as $M_1 = C_1^TC_1$, such that, for a gradient matrix $D$ the following inequality holds for all positive-definite symmetric matrices $\widetilde{M}$:
\begin{equation}
\label{Mopt}
\sqrt{\lambda_{max}\left( (C_1^T)^{-1}D^TM_1DC_1^{-1} \right)} \leq 
\sqrt{\lambda_{max}\left( (\widetilde{C}^T)^{-1}D^T\widetilde{M}D \widetilde{C}^{-1} \right)},
\end{equation}
Below we illustrate how to compute $M_1$ for 2D systems. This can be generalized to higher-dimensional systems.
\paragraph{\rm\textbf{I)~$\mathbf{D}$ has complex conjugate eigenvalues $\mathbf{\lambda}\,\mathbf{=}\,\mathbf{\alpha}\,\mathbf{\pm}\,\mathbf{i}\mathbf{\beta}$}.}
In this case $w\pm iv$ is the associated pair of complex conjugate eigenvectors, where $w,v\in\mathbb{R}^2$. Define: 
\[
C = \begin{bmatrix} w&v\end{bmatrix}^{-1}
\]
As a consequence we have the following equations:
\[
C D C^{-1} = \begin{bmatrix} \alpha&\beta\\-\beta&\alpha\end{bmatrix},\text{ and }(C^T)^{-1}D^TC^T = \begin{bmatrix} \alpha&-\beta\\\beta&\alpha \end{bmatrix},
\]
Thus, one obtains the following results:
\[
(C^T)^{-1}D^TMAC^{-1} = ((C^T)^{-1}D^TC^T)(CDC^{-1}) = \begin{bmatrix}\alpha^2+\beta^2&0\\0&\alpha^2+\beta^2\end{bmatrix},
\]
Clearly, one has that: 
\[
\lambda_{max}((C^T)^{-1}D^TMDC^{-1}) = \alpha^2 + \beta^2.
\]
As the eigenvalues of $(C^T)^{-1}D^TMDC^{-1}$ are equal, it follows from the identity of the determinants that the inequality below holds for any  $\widetilde{M} = \widetilde{C}^T\widetilde{C}$:
\[
\lambda_{max}((C^T)^{-1}D^TMDC^{-1}) \leq \lambda_{max}((\widetilde{C}^T)^{-1}D^T\widetilde{M}D\widetilde{C}^{-1})
\]

\paragraph{\rm\textbf{II) $\mathbf{D}$ has two distinct real eigenvalues $\lambda_1\neq \lambda_2$}}
In this case we do not find a positive-definite symmetric matrix. However, we can find a rotation matrix, defining
coordinates in which the stretching factor is smaller than in canonical coordinates ($M$-norm).
Let $B\in\mathbb{R}^{2\times 2}$ be the eigenvectors matrix of $D$. Denote: 
\[
 B^{-1}DB = \widetilde{D} = \begin{bmatrix} \lambda_1&0\\0&\lambda_2\end{bmatrix}.
\]
%We have
%\[
% B^TD^T(B^T)^{-1}B^{-1}DB = \begin{bmatrix} \lambda_1^2&0\\0&\lambda_2^2\end{bmatrix},\text{ and }\lambda_1\neq\lambda_2.
%\]
%
Let $R$ be the rotation matrix
\[
R = \begin{bmatrix} c & -s\\ s&c\end{bmatrix},\ c,s\neq 0\text{, hence, }
R^{-1}\widetilde{D}R = \begin{bmatrix} \frac{\lambda_1c^2+\lambda_2s^2}{c^2+s^2} & \frac{(\lambda_1-\lambda_2)cs}{c^2+s^2}\\-\frac{(\lambda_1-\lambda_2)cs}{c^2+s^2}  & \frac{\lambda_1c^2+\lambda_2s^2}{c^2+s^2}\end{bmatrix}.
\]
\[
\text{For }s,c = 1\text{ we have }R^T\widetilde{D}^T(R^T)^{-1}R^{-1}\widetilde{D}R = \begin{bmatrix} \left(\frac{\lambda_1+\lambda_2}{2}\right)^2 & 0\\ 0 & \left(\frac{\lambda_1+\lambda_2}{2}\right)^2\end{bmatrix}.
\]
Therefore, we may set 
$
 C = (BR)^{-1},
$
which for $\lambda_1<\lambda_2$ results in\\ $\lambda_{max}(C^TD^T(C^T)^{-1}C^{-1}DC) < \lambda_{max}( D^T D )$, 
because $\left(\frac{\lambda_1+\lambda_2}{2}\right)^2 < \lambda_2^2$ .
\paragraph{\rm\textbf{III) $\mathbf{D\in\mathbb{R}^{n\times n}}$, where $\mathbf{n>2}$}.}
In this case we call the Matlab engine part of our code. Precisely, we use the external linear-optimization packages \cite{mosek,yalmip}. We initially set $\gamma = (\lambda_1\lambda_2\cdots\lambda_n)^{1/n}$. Then, using the optimization package, we try to find $M_1$ and its decomposition, such that: 
\begin{equation}
\label{gamma}
\sqrt{\lambda_{max}\left( (C_1^T)^{-1}D^TM_1DC_1^{-1} \right)}\leq\gamma,
\end{equation}
If we are not successful, we increase $\gamma$ until an $M_1$ satisfying \eqref{gamma} is found.
%
% Experimental evaluation section
\section{Implementation and Experimental Evaluation}
\label{sec:results}
\paragraph{\bf{Prototype Implementation}}
\label{secimpldet}
Our implementation is based on interval arithmetic, i.e. all variables used in the algorithm are over intervals, and all computations performed are executed using interval arithmetic. The main procedure is implemented in C++, which includes header files for the CAPD tool (implemented in C++ as well) to compute rigorous enclosures for the center of the ball at time $t_1$ in step~1, and for the gradient of the flow at time $t_1$ in step~5 of the LRT algorithm (see Section~\ref{sec:rigcomp}). 

To compute the optimal norm and its decomposition for dimensions higher than $2$ in step~2 of the LRT algorithm, we solve a semidefinite optimization problem. We found it convenient to use dedicated Matlab packages for that purpose \cite{mosek,yalmip}, in particular for Case~3 in Section~\ref{secoptnorm}. 

To compute an upper bound $\Lambda$ for the square-root of the maximal eigenvalue of all symmetric matrices in some interval bounds, we used the VERSOFT package \cite{rohn1998bounds,versoft} implemented in Intlab \cite{intlab}. To combine C++ and the Matlab/Intlab part of the implementation, we use an engine that allows one to call Matlab code within C++ using a special makefile. The source code, numerical data, and readme file describing compilation procedure for LRT can be found online~\cite{codes}. 

We remark that the current implementation is a proof of concept; in particular, it is not optimized in terms of the runtime---a direct CAPD implementation is an order of magnitude faster. We will investigate ways of significantly improving the runtime of the implementation in future work.
\vspace{-2ex}
\begin{center}
\begin{table*}
\hspace*{-2ex}
\centering
\bgroup
\def\arraystretch{1.5}% 
\begin{tabular}{|c|c|c|c|c|c|c|c|c|c|c|c|}
%\begin{tabular}{cccccccccccccccc}
\hline % column name
\multirow{2}{*}{BM} & {dt} & {T} & {ID}	&	
      \multicolumn{2}{c|}{LRT}&
      \multicolumn{2}{c|}{Flow*} &
      \multicolumn{2}{c|}{(direct) CAPD} &
      \multicolumn{2}{c|}{nLRT}
      \\
      \cline{5-12}
     &&&& (F/I)V &  (A/I)V & (F/I)V & (A/I)V & (F/I)V & (A/I)V &(F/I)V & (A/I)V\\
	\hline 
%\midrule
%Brusselator
B(2) &0.01 &20 &0.02
&\textbf{7.7e-5} &\textbf{0.09}
&7.9e-5 &0.15
&Fail &Fail
&Fail &Fail\\
\hline
% %Brusselator
% \rowcolor{LightCyan}
% B(2) &0.01 &10 &0.02
% &0.053 &\textbf{26.29}
% &\textbf{0.008} &14.86
% &135.75 &75.24
% &Fail &Fail\\
% \hline
%Van Der Poll
I(2) &0.01 &20 &0.02
&\textbf{4.3e-9} &\textbf{0.09}
&5e-9 &0.12
&7.4e-9 &0.10
&1.45 &1.31
\\
\hline
%Lorentz
L(3) &0.001 &2 &1.4e-6
&1.6e5  &9.0e3
&2.1e13 &7.1e12
&\textbf{4.5e4} &\textbf{1.7e3}
&1.2e22 &4.9e19 
\\
\hline
% %Forced Van Der Poll
% F(2) &0.01 &10 &2e-3
% &\textbf{5.99e-10} &5.57
% &NA &NA
% &0.02 &\textbf{0.70}
% &5.64e7 &3.21e5
% \\
% \hline
%Forced Van Der Poll
%\rowcolor{LightCyan}
F(2) &0.01 &40 &2e-3
&\textbf{1.2e-42} &1.01
&NA &NA
&6.86e-4 &\textbf{0.18}
&5.64e7 &3.21e5
\\
\hline
% %Forced Van Der Poll
% \rowcolor{LightCyan}
% F(2) &0.01 &\textbf{20} &2e-3
% &\textbf{3.05e-23} &2.01
% &NA &NA
% &9.76e-4 &\textbf{0.36}
% &5.64e7 &3.21e5
% \\
% \hline
% %MS Model
% M(2) &1e-3 &1 &2e-3
% &\textbf{0.34} & \textbf{0.63}
% &0.70 &0.88
% &0.69 &0.85
% &1.5 &1.24
% \\
% \hline
% %MS Model
% \rowcolor{LightCyan}
% M(2) &1e-3 &\textbf{2} &2e-3
% &\textbf{0.097} &\textbf{0.412}
% &0.455 &0.721
% &0.45 &0.702
% &1.497 &1.24
% \\
% \hline
%\rowcolor{LightCyan}
M(2) &0.001 &4 &2e-3
&\textbf{0.006} &\textbf{0.22}
&0.29 &0.54
&0.29 &0.52
&4.01 &2.24
\\
\hline
% %Robot Arm Model
% R(4) &0.01 &10 &1e-2
% &\textbf{5.28e-10} &1.13
% &2.18e-7 &0.614
% &5.03e-5 &\textbf{0.46}
% &146.30 &31.59
% \\
% \hline
%Robot Arm Model
%\rowcolor{LightCyan}
R(4) &0.01 &20 &1e-2
&\textbf{2.2e-19} &\textbf{0.07}
&2.4e-15 &0.31
&2.9e-11 &0.23
&Fail &Fail
\\
\hline
% % Biology Model
% O(7) &0.01 &2 &2e-4
% &96.35 &66.30
% &\textbf{91.72} &1.02e5
% &157.88 &\textbf{54.38}
% &Fail &Fail
% \\
% \hline
% Biology Model
O(7) &0.01 &4 &2e-4
&\textbf{71.08} &\textbf{34.35}
&272 &5.1e4
&4.3e3 &620
&Fail &Fail
\\
\hline
%AS Polynomial
P(12))  &5e-4 &0.1 &0.01
&\textbf{5.25} &\textbf{4.76}
&290 &64.6
&280 &62.4
&19.4 &6.2
\\
\hline
%\bottomrule
\end{tabular}
\egroup
\vspace{1ex}
\vspace*{1ex}
\caption{Performance comparison with Flow* and CAPD. 
%For nLRT, the Euclidean norm is used at every time step.
%In the notation B(d), B is the benchmark abbreviation and d is the benchmark dimension. 
We use the following abbreviations: B(2)-Brusselator, I(2)-Inverse Van der Pol oscillator, L(3)-Lorenz attractor, F(2)-Forced Van der Pol oscillator, M(2)-Mitchell Schaeffer cardiac cell model, R(4)-Robot arm, O(7)-Biology model, A(12)-Polynomial system (number inside parenthesis denotes dimension). 
%Due to space constraints, we provide results with only two decimal digits accuracy.  
T: time horizon, dt: time step, ID: initial diameter in each dimension, (F/I)V: ratio of final and initial volume, (A/I)V: ratio of average and initial volume. NA: Not applicable, Fail: Volume blow-up. We mark in bold the best performers.}
\label{table:res}
\vspace{-4ex}
\end{table*}
\end{center}
\vspace*{-4ex}

\vspace{-5ex}
\paragraph{\bf{Experimental Evaluation}}
We compare the results obtained by LRT with direct CAPD and Flow* on a set of standard benchmarks~\cite{CAS2,FKXS}: the Brusselator, inverse-time Van der Pol oscillator, the Lorenz equations~\cite{lorenz1963deterministic}, a robot arm model~\cite{angeli2000characterization}, a 7-dimensional biological model~\cite{FKXS}, and a 12-dimensional polynomial~\cite{anderson2012decomposition} system. Additionally, we consider the forced Van der Pol oscillator~\cite{van1927vii} (a time-variant system), and the Mitchell Schaeffer~\cite{mitchell03} cardiac cell model.

Our results are given in Table~\ref{table:res}, and were obtained on a Ubuntu 14.04 LTS machine, with an Intel Core i7-4770 CPU  3.40GHz x 8  processor and 16 GB memory. The results presented in the columns labeled \emph{(direct) CAPD} and \emph{Flow*} were obtained using the CAPD software package and Flow*, respectively. The internal parameters used in the codes can be checked online \cite{codes}. The comparison metric  that we chose is ratio of the final and initial volume, and ratio of the average and initial volume. We compute directly volumes of the reachtube over-approximations in form of rectangular sets obtained using CAPD, and Flow* software. Whereas, the volumes of the reachtube over-approximations obtained using the LRT algorithm are approximated by the volumes of the tightest rectangular enclosure of the ellipsoidal set, as illustrated on Fig.~\ref{figc}.
The results in \emph{nLRT} column were obtained using a naive implementation of the LRT algorithm, in which the metric space is chosen to be globally Euclidean, i.e., $M_0=M_1=\dots=M_k  = Id$. For MS model the initial condition (i.c.) is in stable regime. For Lorenz i.c. is a period $2$ unstable periodic orbit. For all the other benchmark equations the initial condition was chosen as in \cite{FKXS}.

\begin{figure}
\vspace{-6ex}
\centering
\subfloat[\scriptsize horizon = $\lbrack0,10\rbrack$]
{\includegraphics[height=35mm,width=50mm]{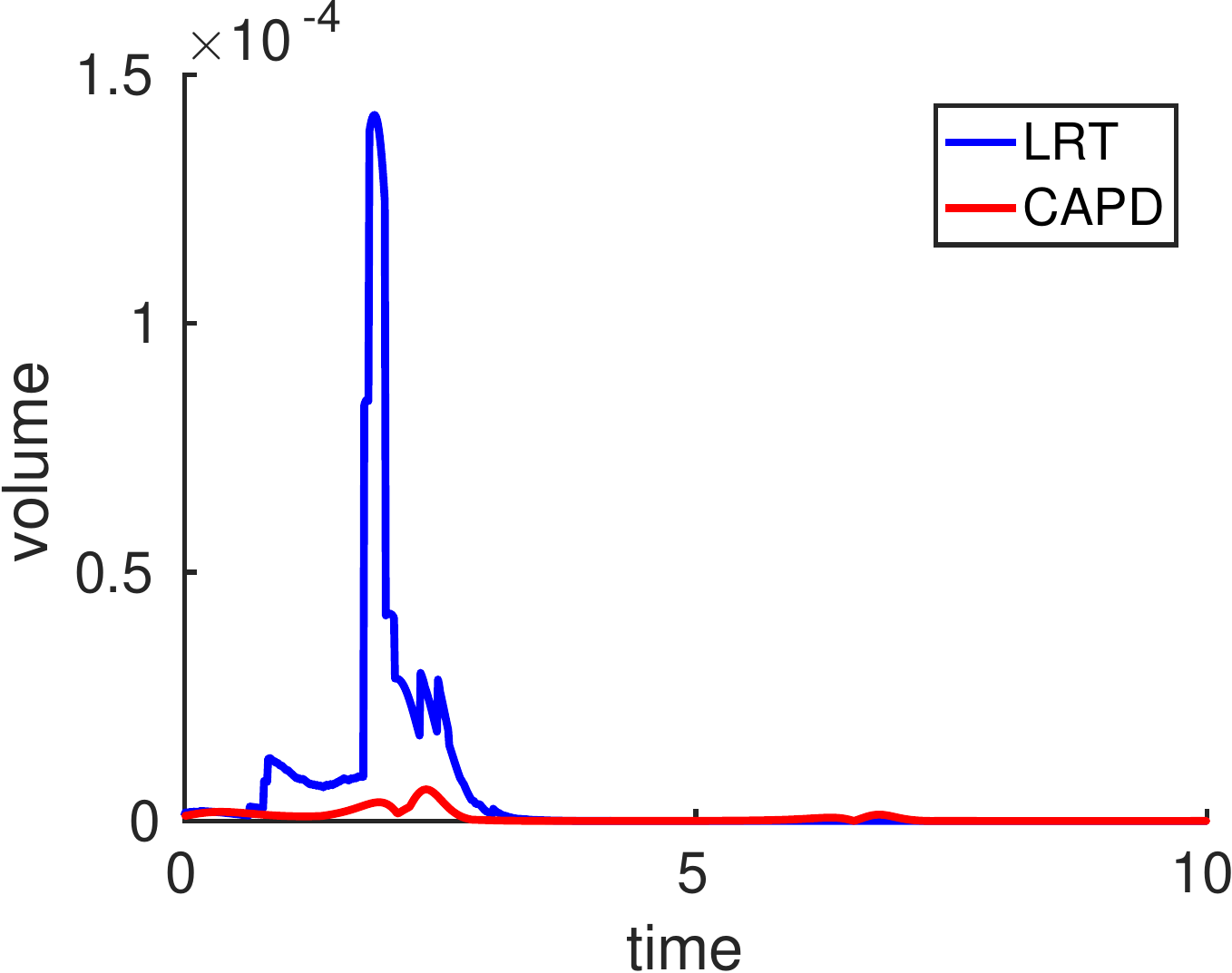}}
\quad \quad 
\subfloat[\scriptsize horizon = $\lbrack3.5, 10\rbrack$]
{\includegraphics[height=35mm,width=50mm]{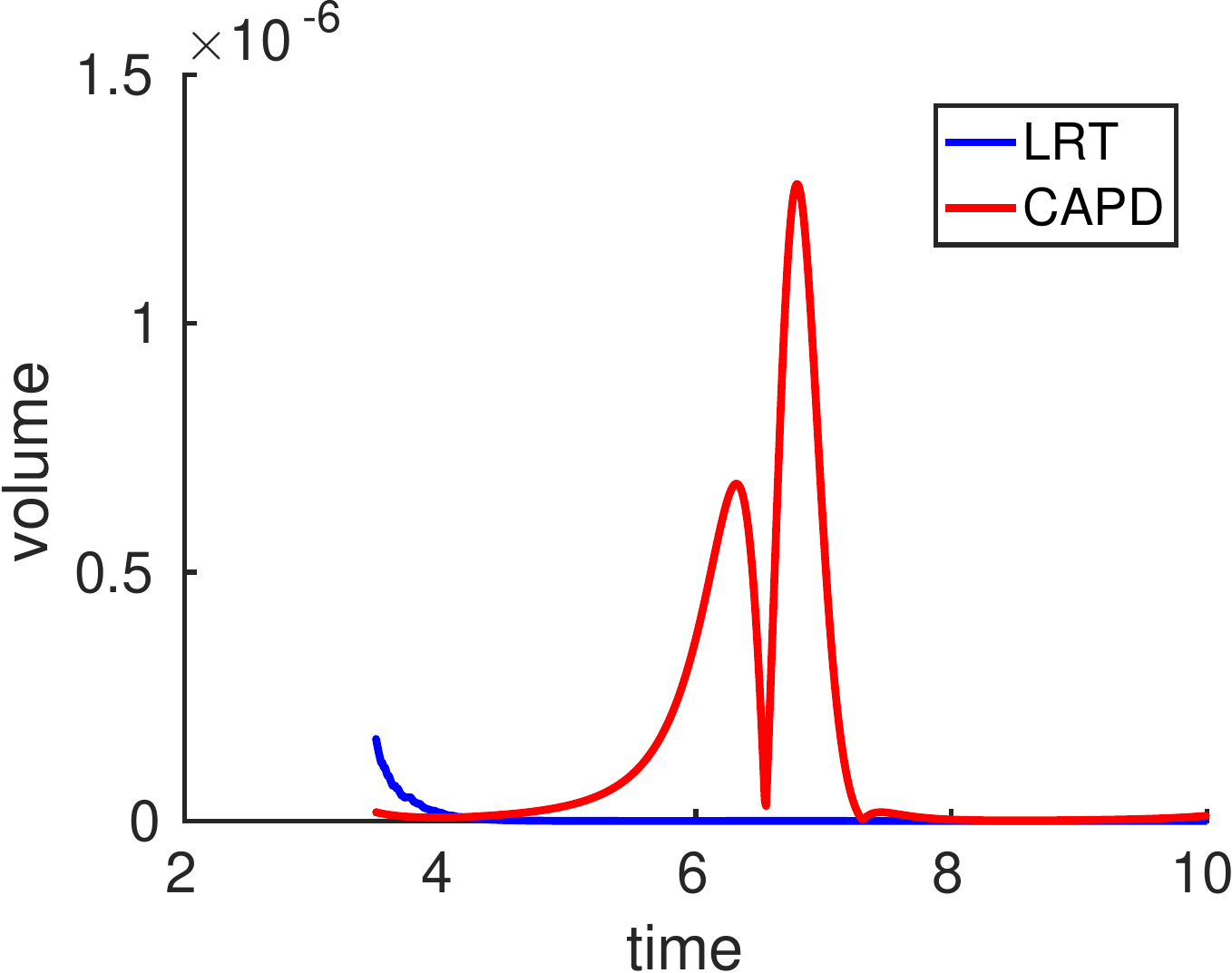}}
\caption{Volume comparison for forced Van der Pol Oscillator}
\label{fig:fvdp}
\vspace{-4ex}
\end{figure}
Some interesting observations about our experiments are as follows. 
Fig.~\ref{fig:fvdp} illustrates that the volumes of the LRT reachtube of the forced Van der Pol oscillator increase significantly for some initial time-steps compared to CAPD reachtube (nevertheless reduce in the long run). This initial increase is related to the computation of new coordinates, which are significantly different from the previous coordinates, resulting in a large stretching factor in step~7 of the LRT algorithm. Namely, we observed that it happens when the coordinates are switched from Case~II to Case~I, as presented in Section~\ref{secoptnorm}. We, however, observe that this does not happen in the system like MS cardiac model (see Fig.~\ref{fig:ms_lorentz}(a)). We believe those large increases of stretching factors in systems like fVDP can be avoided by smarter choices of the norms. We will further investigate those possibilities in future work. 

Table~\ref{table:res} shows that LRT does not perform well for Lorentz attractor compared to CAPD (see also Fig.~\ref{fig:ms_lorentz}(b)), as the CAPD tool is current state of the art for such chaotic systems. LRT, however, performs much better than Flow* for this example. In all other examples, the LRT algorithm behaves favorably (outputs tighter reachtube, compare (F/I)V) in the long run. 

\begin{figure}
\vspace{-6ex}
\centering
\subfloat[\scriptsize  Cardiac Model]
{\includegraphics[height=40mm,width=50mm]{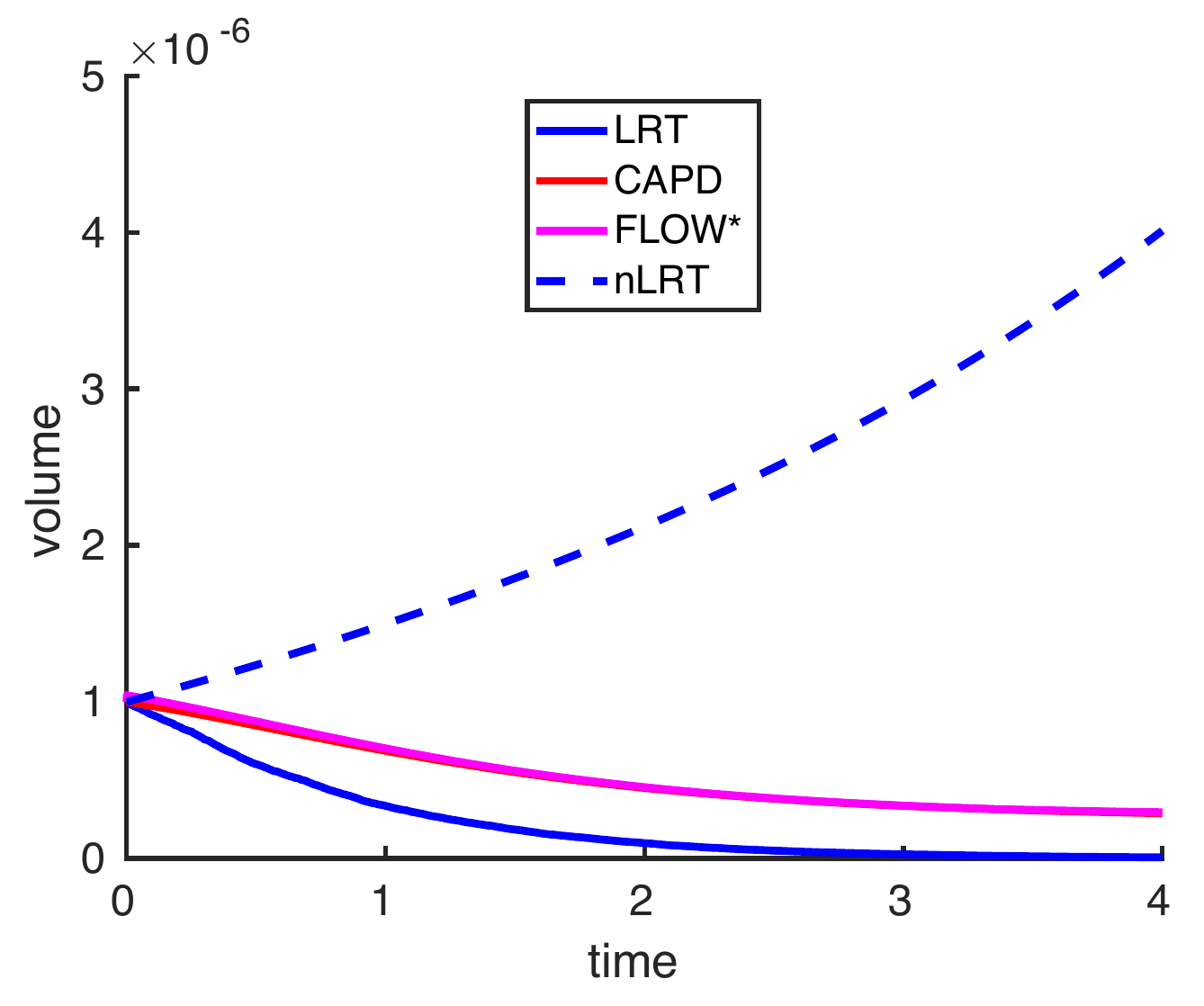}}
\quad \quad 
\subfloat[\scriptsize Lorentz Attractor (Volumes computed by Flow* and nLRT are too large)]
{\includegraphics[height=40mm,width=50mm]{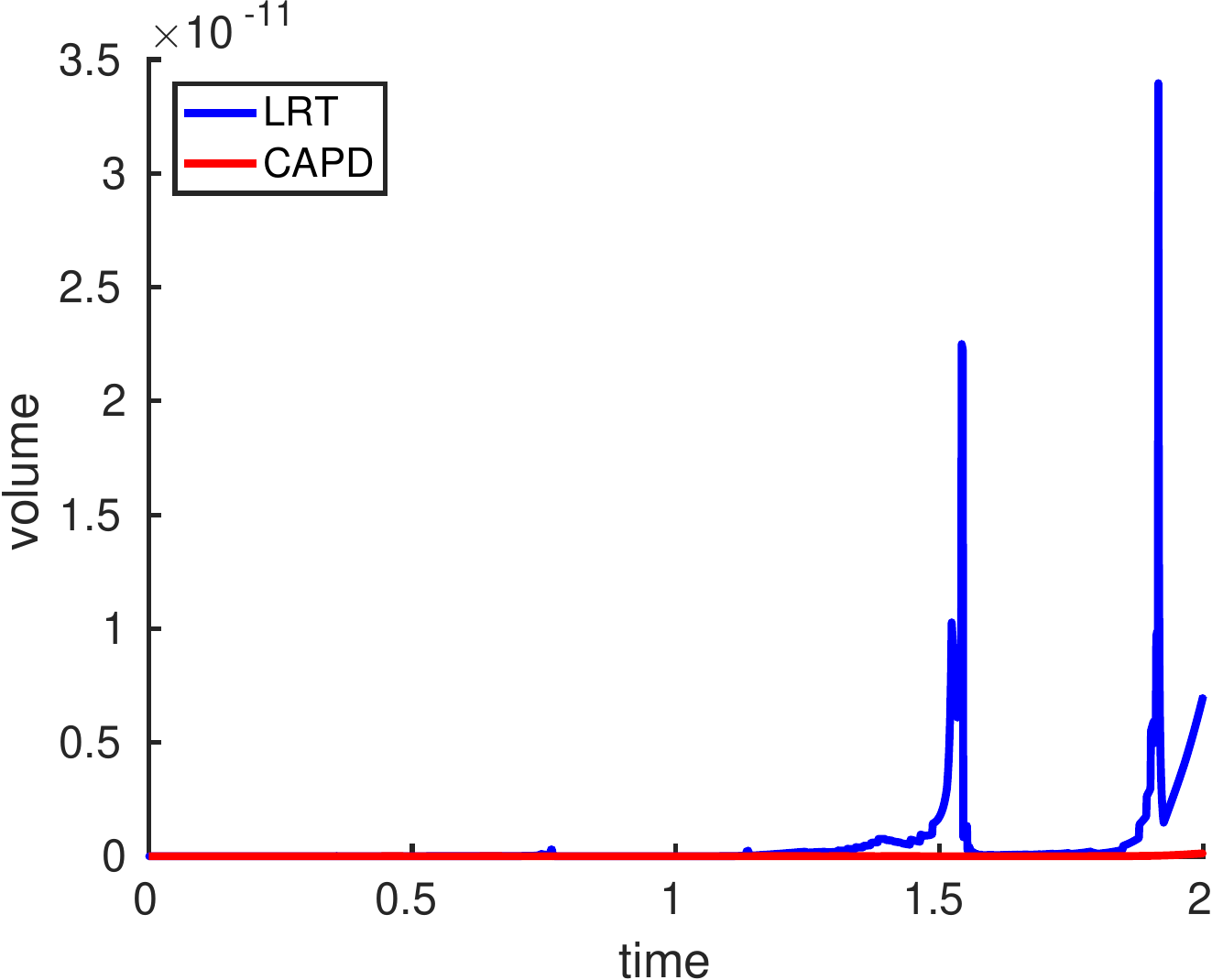}}
\vspace{-1ex}
\caption{Volume comparison in two important benchmarks}
\vspace{-4ex}
\label{fig:ms_lorentz}
\end{figure}

%%It is apparent from the experimental evaluation that of the LRT algorithm performed on ODEs considered as stable -- in as sense that they exhibit a unique attracting orbit that once the reachtube bounds 'settle' on the stable orbit, the bounds become tight (tighter than those obtained from the direct application of the rigorous tools). 

%%%The  spikes visible for the LRT curve in Figures~\ref{fig:fvdp} and~\ref{fig:ms_lorentz}~(b) are due to some inappropriate norm switches in step~2 of the LRT algorithm (see Section~\ref{sec:rigcomp}). This may negatively affect stability of our algorithm, and we plan to work further on the norm selection part of the algorithm,

%\end{comment}

%% \vspace*{-3ex}

\vspace{-4ex}
\vspace{2ex}
\section{Conclusions}
\label{sec:conclusions}
\vspace{-2ex}
We presented LRT, a rigorous procedure for computing ReachTube overapproximations based on Cauchy-Green stretching factor computation, in Lagrangian coordinates. We plan to pursue further research on our algorithm. One appealing possibility is to extend LRT to hybrid systems widely used in research on cardiac dynamics, and many other fields of study. 

We also plan to implement LRT with forward-in-time integration and con\-serv\-a\-tive-enclosures computation of the gradient, by just using a simple independent code (instead of CAPD). By such a code we mean a rigorous integration procedure based, for example, on the Taylor's method. This code would directly compute an interval enclosure for the Cauchy-Green deformation tensors in an appropriate metric space. It would then be interesting to compare the wrapping-reduction performance of such a code with the procedure described above.

% 
%\nocite{*}
%
\paragraph{Acknowledgments}
\vspace{-2ex}
Research supported in part by the following 
grants: NSF IIS-1447549, NSF CPS-1446832, NSF CPS-1446725, 
NSF CNS-1445770, NSF CNS-1430010, AFOSR FA9550-14-1-0261
and ONR N00014-13-1-0090.  J. Cyranka was supported in part by NSF grant DMS 1125174 and DARPA contract HR0011-16-2-0033.
\vspace{-2ex}
\bibliography{CG_algorithm}
\bibliographystyle{abbrv}
\newpage
\appendix
\section{Proofs of the Lemmas in Section~\ref{sec:bgd}}
\label{app:proofs}
\begin{lemma}
Consider the Cauchy problem \eqref{cauchy}. Let $x_0,y_0\in\mathbb{R}^n$ be two initial conditions at time $t_0$.
Let $M\in\mathbb{R}^{n\times n}$ be a positive-definite symmetric matrix and $C^TC = M$ be its decomposition. For $t_1 \geq t_0$, it holds that 
\begin{equation*}
%\label{Mestim}
  \|\phi_{t_0}^{t_1}(x_0){-}\phi_{t_0}^{t_1}(y_0)\|_M\,{\leq}\,\sqrt{\hat{\lambda}\left((C^T)^{-1}\left(D_x\phi_{t_0}^{t_1}(\xi)\right)^T 
  M\,D_x\phi_{t_0}^{t_1}(\xi) C^{-1}\right)}~\|x_0-y_0\|_M
\end{equation*}
where $\xi = \omega x_0 + (1-\omega)y_0$ for some $\omega\in[0,1]$.
For the particular case of the Euclidean norm, \eqref{Mestim} takes the form
\begin{equation*}
%\label{estim}
\|\phi_{t_0}^{t_1}(x_0)-\phi_{t_0}^{t_1}(y_0)\|_2 \leq \sqrt{\lambda_{max}\left(\left(D_x\phi_{t_0}^{t_1}(\xi)\right)^T\cdot D_x\phi_{t_0}^{t_1}(\xi)\right)}~\|x_0-y_0\|_2.
\end{equation*}
\end{lemma}
\begin{proof}
Let $\xi(\omega) = \omega x_0 + (1 - \omega)y_0$.
From 
\[
\int_0^1{D_x\phi_{t_0}^{t_1}( \xi(\omega) )\,d\omega} = \frac{1}{x_0-y_0}\left(\phi_{t_0}^{t_1}(x_0) - \phi_{t_0}^{t_1}(y_0)\right),
\]
and the well known mean value theorem for integrals, it holds that
\begin{equation*}
%\label{mvt}
\phi_{t_0}^{t_1}(x_0)-\phi_{t_0}^{t_1}(y_0) = D_x\phi_{t_0}^{t_1}(\hat{\xi})(x_0-y_0)
\end{equation*}
for some $\hat{\omega}\in[0,1]$, $\xi(\hat{\omega}) = \hat{\xi}$.
%We are interested in computing $\rd(T)$, i.e. an upper bound for the radius of the ball, such that it encloses all trajectories at time $T$
%\[
%\rd(T) \geq \sup_{y\in B(x_0,\rd(t_0))}{\|\phi_{t_0}^T(y) - \phi_{t_0}^T(x)\|}.
%\]
From taking norms in above equation we obtain
\[
{\|\phi_{t_0}^{t_1}(x_0) - \phi_{t_0}^{t_1}(y_0)\|} = \|D_x\phi_{t_0}^{t_1}(\hat{\xi})(x_0-y_0)\|\leq \|D_x\phi_{t_0}^{t_1}(\hat{\xi})\|\|x_0-y_0\|.
\]
Replacing $\|D_x\phi_{t_0}^{t_1}(\hat{\xi})\|_2$ with the induced Euclidean matrix norm we obtain \eqref{estim}. If we use the weighted $M$-norm~\eqref{Mnorm}
we have for the matrix norm \eqref{MnormMatrix}
\begin{equation*}
  \|\phi_{t_0}^{t_1}(x_0){-}\phi_{t_0}^{t_1}(y_0)\|_M\,{\leq}\,\sqrt{\lambda_{max}((C^T)^{-1}\left(D_x\phi_{t_0}^{t_1}(\hat{\xi})\right)^T 
  M\,D_x\phi_{t_0}^{t_1}(\hat{\xi}) C^{-1})}~\|x_0-y_0\|_M
\end{equation*}
\qed
\end{proof}
\begin{lemma}
\label{lemdiscrep2}
Consider the Cauchy problem \eqref{cauchy}. Let $x_0,y_0\in\mathbb{R}^n$ be two initial conditions at time $t_0$.
Let $M_0,M_1\in\mathbb{R}^{n\times n}$ be positive-definite symmetric matrices, and $C_0^TC_0 = M_0$, $C_1^TC_1 = M_1$ their decompositions
respectively. For $t_1 \geq t_0$, it holds that 
\begin{multline*}
  \|\phi_{t_0}^{t_1}(x_0)-\phi_{t_0}^{t_1}(y_0)\|_{M_1}\\\leq \sqrt{\lambda_{max}\left((C_0^T)^{-1}\cdot\left(D_x\phi_{t_0}^{t_1}(\xi)\right)^T\cdot M_1\cdot D_x\phi_{t_0}^{t_1}(\xi)\cdot C_0^{-1}\right)}~\|x_0-y_0\|_{M_0},
\end{multline*}
where $\xi = \omega x_0 + (1-\omega)y_0$ for some $\omega\in[0,1]$.
\end{lemma}
\begin{proof}
Let $\xi = \omega x_0 + (1-\omega)y_0$ for some $\omega\in[0,1]$. We use the equality $\phi_{t_0}^{t_1}(x_0)-\phi_{t_0}^{t_1}(y_0) = D_x\phi_{t_0}^{t_1}(\xi)(x_0-y_0)$ derived in the proof of Lemma~\ref{lemdiscrep}.
Let us denote $A:=D_x\phi_{t_0}^{t_1}(\xi)$, and $w = (x_0-y_0)$. It holds that
\begin{multline*}
\|\phi_{t_0}^{t_1}(x_0)-\phi_{t_0}^{t_1}(y_0)\|_{M_1} = \|Aw\|_{M_1}=\sqrt{(Aw)^TM_1(Aw)}=\sqrt{w^T(A^TM_1A)w}=\\
\sqrt{w^TC_0^T((C_0^T)^{-1}A^TM_1AC_0^{-1})C_0w} \leq \sqrt{\lambda_{max}\left((C_0^T)^{-1}A^TM_1AC_0^{-1}\right)}\sqrt{w^TC_0^TC_0w}\\
\hspace*{-5mm}=\sqrt{\lambda_{max}\left((C_0^T)^{-1}A^TM_1AC_0^{-1}\right)}~\|w\|_{M_0}\quad\qed
\end{multline*}
\end{proof}

\end{document}